\documentclass[lettersize,journal]{IEEEtran}
\IEEEoverridecommandlockouts
\usepackage{cite}
\usepackage{amsmath,amssymb,amsfonts}
\usepackage[english]{babel}
\usepackage{amsthm}

\usepackage{dblfloatfix} 
\usepackage{algorithm}
\usepackage{algpseudocode}
\usepackage[cjk]{kotex}
\usepackage{graphicx}

\usepackage{graphicx}
\usepackage{textcomp}
\usepackage{xcolor}
\usepackage{subfigure}
\usepackage{tabulary}
\usepackage{booktabs}
\usepackage{tabularx}
\usepackage{array}
\newcolumntype{L}[1]{>{\raggedright\let\newline\\\arraybackslash\hspace{0pt}}m{#1}}
\newcolumntype{C}[1]{>{\centering\let\newline\\\arraybackslash\hspace{0pt}}m{#1}}
\newcolumntype{R}[1]{>{\raggedleft\let\newline\\\arraybackslash\hspace{0pt}}m{#1}}
\usepackage{multirow}
\usepackage{mathtools}

\newtheorem{thm}{\textbf{Theorem}}

\newtheorem{lem}{\textbf{Lemma}}
\newtheorem{prop}{\textbf{Proposition}}[thm]

\usepackage{lipsum} 
\usepackage{xcolor}

\def\BibTeX{{\rm B\kern-.05em{\sc i\kern-.025em b}\kern-.08em
    T\kern-.1667em\lower.7ex\hbox{E}\kern-.125emX}}
\begin{document}

\title{Low Complexity Frequency Domain Nonlinear\\ Self-Interference Cancellation for Flexible Duplex}

\author{Yonghwi Kim,~\IEEEmembership{Graduate Student Member, IEEE}, Kai-Kit Wong,~\IEEEmembership{Fellow, IEEE},\\ Jianzhong (Charlie) Zhang,~\IEEEmembership{Fellow, IEEE}, and Chan-Byoung Chae~\IEEEmembership{Fellow, IEEE}

\thanks{Y. Kim and C.-B. Chae are with the School of Integrated Technology, Yonsei University, Seoul, 03722, South Korea (e-mail: eric\_kim@yonsei.ac.kr; cbchae@yonsei.ac.kr).}%
\thanks{K. Wong is with the Department of Electronic and Electrical Engineering, University College London, U.K. (e-mail: kai-kit.wong@ucl.ac.uk). He is also affiliated with Yonsei Frontier Lab., Yonsei University, Seoul, 03722, South Korea (e-mail: kai-kit.wong@ucl.ac.uk).}
\thanks{J. Zhang is with the Advanced Communications Research Center, Samsung Research, Seoul, South Korea.  He is also with the Samsung Research America, Plano, TX, 75023 USA (e-mail: jianzhong.z@samsung.com).}

\thanks{Manuscript received XXX, XX, 2024; revised XXX, XX, 2024.}}
\markboth{IEEE Journal,~Vol.~XX, No.~XX, XXX~2024}{}
\maketitle

\begin{abstract}     

Nonlinear self-interference (SI) cancellation is essential for mitigating the impact of transmitter-side nonlinearity on overall SI cancellation performance in flexible duplex systems, including in-band full-duplex (IBFD) and sub-band full-duplex (SBFD). Digital SI cancellation (SIC) must address the nonlinearity in the power amplifier (PA) and the in-phase/quadrature-phase (IQ) imbalance from up/down converters at the base station (BS), in addition to analog SIC. In environments with rich signal reflection paths, however, the required number of delayed taps for time-domain nonlinear SI cancellation increases exponentially with the number of multipaths, leading to excessive complexity. This paper introduces a novel, low-complexity, frequency domain nonlinear SIC, suitable for flexible duplex systems with multiple-input and multiple-output (MIMO) configurations. The key approach involves decomposing nonlinear SI into a nonlinear basis and categorizing them based on their effectiveness across any flexible duplex setting. The proposed algorithm is founded on our analytical results of intermodulation distortion (IMD) in the frequency domain and utilizes a specialized pilot sequence. This algorithm is directly applicable to orthogonal frequency division multiplexing (OFDM) multi-carrier systems and offers lower complexity than conventional digital SIC methods. Additionally, we assess the impact of the proposed SIC on flexible duplex systems through system-level simulation (SLS) using 3D ray-tracing and proof-of-concept (PoC) measurement.
\end{abstract}

\begin{IEEEkeywords}
full-duplex (FD), flexible duplex, and self-interference cancellation (SIC).
\end{IEEEkeywords}

\section{Introduction}

\IEEEPARstart{F}{lexible} duplex, including in-band full-duplex (IBFD) and sub-band full-duplex (SBFD), such as cross-division duplex (XDD), is emerging to meet the demands of beyond 5th generation (B5G) and 6th generation (6G) cellular networks' requirements~\cite{FD6G,advD,GuestChae,3gppDupEn,kim2015survey,FDnet_hwi}. These systems provide flexible uplink (UL) and downlink (DL) resource allocation. This helps to overcome the performance limitations of traditional duplex technologies like time-division duplex (TDD) and frequency-division duplex (FDD), which have issues such as guard-band configuration and switching latency~\cite{alpha1,US-alpha,flex1}.

 {In 5G communications, significant path loss challenges cell coverage. Researchers are addressing this by exploring duplex enhancement~\cite{3gppDupEn} and beamforming technologies~\cite{CBF, RISNOMA_hwi}. Effective resource allocation is crucial to improving throughput, the cell coverage, and fairness among users~\cite{FDnet_hwi}. Recently, the standard has defined duplex enhancement in terms of SBFD operation and is actively addressing it as a work item within the 3rd generation partnership project (3GPP)~\cite{3gppDupEn}. Unlike the previous papers that dealt with static subcarrier allocation, a recent document~\cite{3gppDupEn} discusses dynamic SBFD. This means that UL/DL subcarrier allocation can be changed depending on the situation, and SIC technology needs to be advanced accordingly. IBFD systems fully overlap UL and DL spectrum to achieve maximum spectral efficiency~\cite{FD,FDchae}. SBFD systems, like XDD, adjust the UL spectrum ratio to extend the coverage of base stations (BS) in cellular networks~\cite{XDD}. As a next-generation technology, flexible duplex enhances both throughput and coverage.}

Managing UL/DL in IBFD and SBFD requires both analog and digital self-interference cancellation (SIC) methods. Analog SIC aims to prevent radio frequency (RF) chain saturation by reducing SI to manageable levels before it reaches the analog-to-digital converter (ADC)~\cite{analSIC}, employing techniques such as active circuits~\cite{ASICjw,ASICjsac} or beamforming cancellation (BFC)~\cite{IanBSI,IanMag}. Digital SIC compensates for RF chain nonlinearity, addressing issues like power amplifier (PA) distortions and mixer imbalances at the baseband~\cite{NLSI}.

Operating nonlinear digital SIC presents challenges due to its complexity~\cite{NLSI,IanMag,K_iter}. However, considering various crucial aspects for flexible duplex systems in the frequency domain can reduce this complexity. As discussed in Section~IV, nonlinearity behaves differently depending on the UL/DL subcarrier allocation, which operates in various duplex modes. The integration of passive analog SIC and BFC modifies the SI channel behavior, aiding in the suppression of less critical nonlinearities, especially in terms of frequency selectivity~\cite{analSIC,IanFreq}. This paper proposes a low-complexity method for processing digital SIC in the frequency domain, incorporating these diverse factors to enhance the overall efficiency of flexible duplex systems.


\subsection{Related Works on Frequency Domain Nonlinear SIC}

Digital SIC addresses RF chain impairments and the multipath nature of the SI channel, introducing frequency selectivity~\cite{NLSI,NLSIsim}. SBFD scenarios, not featuring UL/DL overlap, may still experience out-of-band emission (OOBE) as SI due to PA nonlinearity~\cite{XDD}. The aggregate SIC performance, achieving metrics like $-110$~dB cancellation, integrates both passive and active analog SICs with digital nonlinear SIC.

Nonlinear SI mitigation strategies comprise active analog SIC and digital SIC. Active analog SIC, which estimates delay taps for circuit implementation, might introduce its own impairments~\cite{ASICjw,ASICjsac,ASICfreq}. Digital SIC, depend on proper ADC operation at the receiver, accurately addresses nonlinearity by processing baseband signals~\cite{IanMag,NLSIsim}. Effective implementation of analog SIC, including BFC~\cite{IanBSI}, allows digital SIC to minimize residual SI to Rx noise levels.

Digital nonlinear SIC traditionally uses time domain approaches, demanding calculations across all sampling points and leading to significant complexity~\cite{K_FH}. Early research in digital SIC, addressing PA nonlinearity with the Parallel-Hammerstein (PH) model, primarily focus on time domain SIC~\cite{FDRs,NLSI}. Subsequent studies extended these models to account for IQ imbalance~\cite{K_iter,WideLinFD} and phase noise~\cite{SIC_PN,SIC_PN2}. 

Frequency domain digital SIC, a solution to reduce complexity, targeted each OFDM subcarrier specifically~\cite{K_FH}. The authors of \cite{compactFDd2d} demonstrated frequency domain SIC through PoC, with a focus on linear components. For further complexity reduction, nonlinear basis selection for SI coefficient estimation have been proposed in a heuristic manner~\cite{K_basis}. Additionally, \cite{FDC_iter} introduced a successive channel estimation method to reduce complexity than \cite{K_basis}.

Previous studies have focused on single FD mode, primarily addressing IBFD scenarios, with references such as \cite{XDD,202212access,SBFD_hwi} exploring SBFD scenarios. However, there has been limited exploration in flexible duplex configurations that enable dynamic allocation of UL/DL subcarriers. Low-complexity strategies~\cite{K_basis,FDC_iter} have predominantly focused on IBFD, relying on random training signals rather than a closed-form approach for flexible duplex scenarios.

\subsection{Contributions and Organization of This Paper}

\begin{table}[]
\centering
\caption{Comparison for FD digital SIC studies}
{%
\begin{tabular}{c|c|cc|cc}
\toprule
\multirow{2}{*}{\textbf{Ref.}} & \multirow{2}{*}{Duplex Mode} & \multicolumn{2}{c|}{Nonlinearity}                                          & \multicolumn{2}{c}{System Design}                                          \\ \cline{3-6} 
                       &                              & \multicolumn{1}{c|}{PA}           & IQ Imb.              & \multicolumn{1}{c|}{PoC}         & Pilot              \\ \hline
\cite{XDD}                    & SBFD                         & \multicolumn{1}{c|}{\checkmark} & -                         & \multicolumn{1}{c|}{\checkmark} & -                         \\
\cite{compactFDd2d}           & IBFD                         & \multicolumn{1}{c|}{-}                         & -                         & \multicolumn{1}{c|}{\checkmark} & \checkmark \\
\cite{K_basis}               & IBFD                         & \multicolumn{1}{c|}{\checkmark} & \checkmark & \multicolumn{1}{c|}{-}                         & \checkmark \\
\cite{FDC_iter}              & IBFD                         & \multicolumn{1}{c|}{\checkmark} & -                         & \multicolumn{1}{c|}{-}                         & -                         \\
This Work              & Flexible (All)               & \multicolumn{1}{c|}{\checkmark} & \checkmark & \multicolumn{1}{c|}{\checkmark} & \checkmark \\ \bottomrule
\end{tabular}%
\label{table.works}
}
\end{table}

While researchers have explored handling nonlinear SI in the frequency domain~\cite{K_basis,FDC_iter}, there has been limited exploration into running real-time digital SIC operations. As shown in Table~\ref{table.works}, previous studies have predominantly focused on fixed IBFD scenarios, neglecting the need to adapt to changes in UL/DL subcarrier allocations. Digital SIC effectively manages the variable nonlinearity resulting from such dynamic allocations. Although incorporating frequency-selective SI channels and UL/DL subcarrier allocations, certain nonlinear bases induced by PA nonlinearity and IQ imbalance do not significantly impact the efficacy of digital SIC, despite their computational complexity. Additionally, the use of certain training pilots for digital SIC, as described in~\cite{compactFDd2d,K_basis}, have yet to fully utilize the distinct properties of nonlinear SI, indicating a gap in leveraging the full potential of digital SIC strategies.

\begin{figure*}

	\begin{center}
		{\includegraphics[width=2.0\columnwidth,keepaspectratio]
			{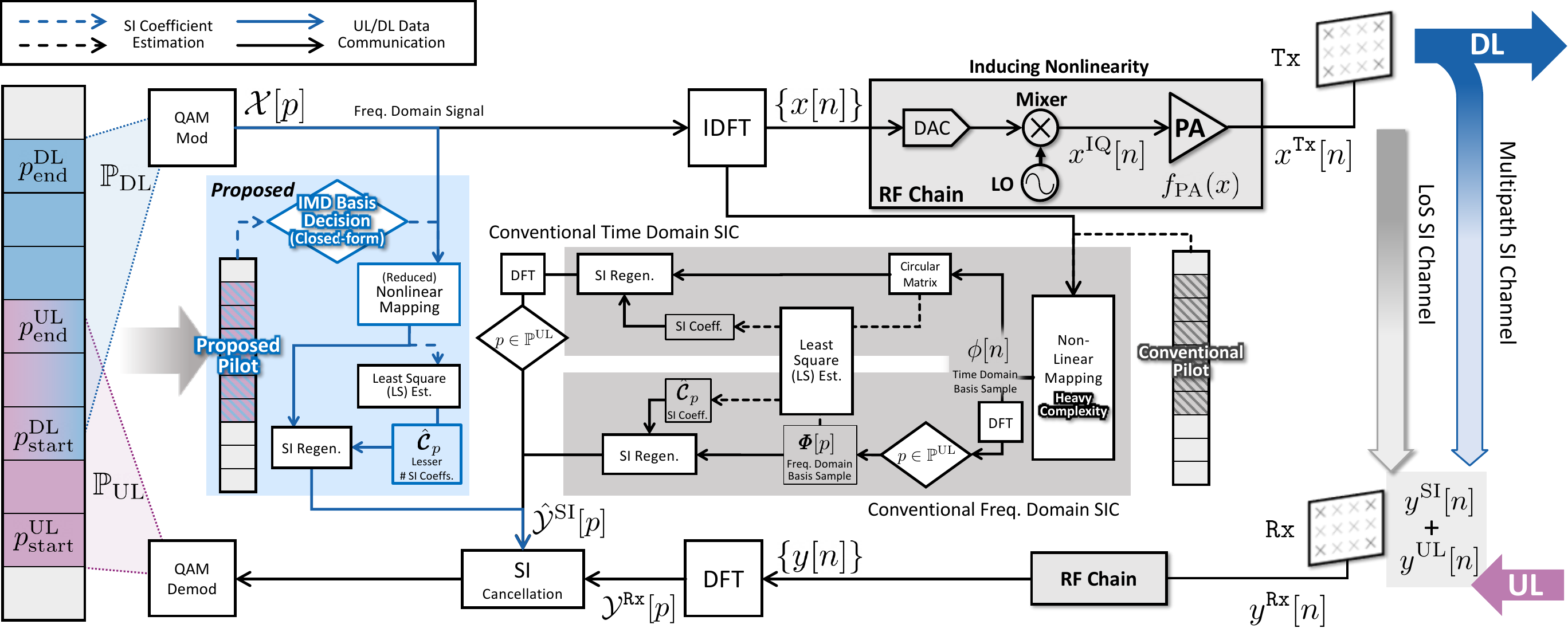}%
			\caption{The system model and assumptions: the procedure of nonlinear SIC for both the proposed method and the conventional method.}
			\label{fig.sysmodel}
		}
	\end{center}
\vspace{-15pt}
\end{figure*}

To address these issues, we propose a low-complexity frequency-domain nonlinear SIC method suitable for flexible duplex systems. The proposed method offers a frequency-domain SIC that is compatible with OFDM-MIMO systems. By using closed-form solutions and effective pilot symbols, our approach minimizes the computational steps required for running SIC and estimating SI coefficients.


\begin{itemize}

\item We conduct a theoretical analysis of nonlinear SI in the frequency domain, using theorems and propositions. We focus on intermodulation distortion (IMD) caused by PA and IQ mixers, examining its variations across different UL/DL configurations in flexible duplex systems.

\item We propose an algorithm and pilot signal designed to reduce the complexity of comprehensive SIC. This is achieved by distinguishing between the SI channel, PA nonlinearity, and IQ imbalance, leveraging the long-term characteristics. By enhancing the latest digital SIC research with the proposed closed-form solutions, which improves upon heuristic approaches in~\cite{K_basis,FDC_iter}, the proposed digital SIC technology further reduces complexity.

\item We evaluate the impact of the proposed frequency domain, low-complexity, nonlinear SIC at the network-level. The nonlinearity caused by PA was measured and applied through a software-defined radio (SDR) PoC. The frequency-selectivity of the SI channel and the multi-user UL/DL channel are measured by 3D ray-tracing to reflect the multipath characteristics in real-world environments.
\end{itemize}

This paper is organized as follows. In Section~II, we introduce the preliminaries on SIC and the flexible duplex to describe our system model and assumptions. In Section~III, we formulate our complexity reduction problem of the nonlinear digital SIC. In Section~IV, we explain the details of the proposed SIC methods based on lemmas, theorems and propositions. Simulation results and our conclusion are given in Sections V and VI.

\emph{Notations}: In this paper, $A$ and $a$ are scalars and $\mathbf{A}$ and $\mathbf{a}$ are both vector. $\mathcal{A}[p]$ and $\pmb{\mathcal{A}}[p]$ are in the frequency domain, $a[n]$ and $\mathbf{a}$ are in the time domain. $\mathbf{A}^\top, \mathbf{A}^\mathsf{H}$ denote the transpose and Hermitian of $\mathbf{A}$, respectively. $|\cdot|$ denotes the modulus of a complex number or the cardinality of a set. $\mathbb{C}$ denotes the complex number set, and $\mathbb{C}^{M\times N}$ denotes the space of $M\times N$ complex-valued matrices. $j\triangleq\sqrt{-1}$ is the imaginary unit. 

\section{System Model and Assumptions}

\subsection{Overall Flexible Duplex System}

We consider a flexible duplex scenario illustrated in Fig.~\ref{fig.sysmodel}. A BS operates an OFDM system with the number of subcarrier as $P$ and the subcarriers indexed by $p$, after sampled by subcarrier spacing, $\Delta f$. In the time domain, an OFDM symbol is sampled to index $n$ among $N$ samples, and each symbol index by $m$ among $M$ symbols. The symbols for SI estimation are represented as $m\in\mathbb{M}^\mathtt{train}$ and symbols $m\in\mathbb{M}^\mathtt{run}$ for the data communication that means running SIC. The discrete Fourier transform (DFT) and the inverse DFT (IDFT) operation transforms the domain of the signal between the time domain and the frequency domain.

The BS $\mathtt{Tx}$ transmit the DL signal, $\mathcal{X}_m[p]$ to DL-user equipment (UE). The BS $\mathtt{Rx}$, simultaneously, decodes the UL signal $\mathcal{Y}_m^\text{UL}[p]$ from UL-UE. The OFDM system transforms DL signal from $\mathcal{X}_m[p]$ to $x^\text{DL}_m[n]$, and UL signal from $y_m^\text{UL}[n]$ to $\mathcal{Y}_m^\text{UL}[p]$. In FD operating scenario, $\mathcal{X}_m[p]$ generates the SI as $\mathcal{Y}^\text{SI}_m[p]$, the overall received signal at $\mathtt{Rx}, \mathcal{Y}_m^\mathtt{Rx}[p]$ is as follows:
\begin{equation}
\begin{aligned}
\mathcal{Y}_m^\mathtt{Rx}[p]&= \mathcal{Y}_m^\text{UL}[p] + \mathcal{Y}_m^\text{SI}[p] + \mathcal{Z}[p],
\end{aligned}
\end{equation}
where $\mathcal{Z}[p]$ is the complex Gaussian noise at $\mathtt{Rx}$. The digital SIC targets to estmiate the SI-only received signal as $\hat{\mathcal{Y}}_m^\text{SI}$ in terms of $\mathcal{X}_m$ and cancel it out as $\mathcal{Y}^\text{UL}_m=\mathcal{Y}^\mathtt{Rx}_m-\hat{\mathcal{Y}}^\text{SI}_m$. 

\subsubsection{Subcarrier flexibility in frequency domain}
To model the flexibility of UL/DL allocation, the set of the system subcarrier, $\mathbb{P}$, is as follows:
\begin{equation}
\mathbb{P}=\{1,\cdots,p,\cdots,P\}.
\end{equation}
The system allocates DL and UL subcarrier set, $\mathbb{P}^\text{DL},\mathbb{P}^\text{UL}\subset\mathbb{P}$, as follows:
\begin{equation}
\begin{aligned}
\mathbb{P}^\text{DL}=\left\{p^\text{DL}_\text{start}, \cdots, p^\text{DL}_\text{end}\right\}, \quad \mathbb{P}^\text{UL}=\left\{p^\text{UL}_\text{start}, \cdots, p^\text{UL}_\text{end}\right\}, 
\label{eq.pULDL}
\end{aligned}
\end{equation}
where $p^\text{DL}_\text{start}$ and $p^\text{DL}_\text{end}$ represent the boundary indices for DL, while $p^\text{UL}_\text{start}$ and $p^\text{UL}_\text{end}$ correspond to those for UL. For DL, $\mathcal{X}_m[p_i]=0\quad \forall p_i\notin\mathbb{P}_\text{DL}$ holds.
When $\mathbb{P}^\text{UL}=\mathbb{P}^\text{DL}$ holds, the duplex mode is IBFD, while SBFD when $\mathbb{P}^\text{UL} \cap \mathbb{P}^\text{DL}=\varnothing$ holds. The flexible duplex exclusively supports the special duplex mode: partially overlapping UL and DL spectrum as $\mathbb{P}^\text{UL} \cap \mathbb{P}^\text{DL}=\mathbb{P}^\text{OV}$, where $\mathbb{P}^\text{OV}$ is set of subcarriers allocated for both UL and DL.

$\mathtt{Tx}$ transmits the signal, and $\mathcal{X}_m[p]$ satisfies the transmit power constraint as $E \left[\left|\mathcal{X}_m[p]\right|^2\right]=A_\text{digi}^2$, to meet the transmit power constraint in the digital domain. Before passing RF chain of $\mathtt{Tx}$, the cyclic prefix (CP) is appended in front of $x[n]$ by $N_\text{CP}$ as follows:
\begin{equation}
\begin{aligned}
x_m=\mathtt{IDFT}\left\{\mathcal{X}_m\right\}
\rightarrow x_m^\text{CP}&=\left[x_m[N-N_\text{CP}:N], x_m\right].
\end{aligned}
\end{equation}


\subsubsection{IQ imbalance from mixer} 

Prior to reaching the PA, the downlink baseband signal, $x_m$, passes through a local oscillator and mixer, introducing IQ imbalance~\cite{IQ_tsp,WideLinFD}. This imbalance is effectively captured by the signal $x^\text{IQ}_m (\mathcal{X}^\text{IQ}_m)$~\cite{K_basis,K_iter}, which can be represented using the IQ imbalance coefficient $b^\text{IQ}$ in both the frequency and time domains as follows:
\begin{equation}
\begin{aligned}
x^\text{IQ}_m[n]&=x_m[n]+b^\text{IQ}(x_m[n])^*\\
\rightarrow \mathcal{X}_m^\text{IQ}[p]&=\mathcal{X}_m[p]+b^\text{IQ}(\mathcal{X}_m[-p])^*.
\end{aligned}
\end{equation}
The image rejection rate (IRR), which serves as a measure of IQ imbalance, is quantified by $1/|b^\text{IQ}|^2$.

\subsubsection{PA nonlinearity} 
PA amplifies signal to target transmit power as follows:  
\begin{equation}
x_m^\mathtt{Tx}[n]=f_\text{PA}\left(x_m^\text{IQ}[n]\right),
\label{eq.xtx}
\end{equation}
where$f_\text{PA}$ is amplifying function of PA in analog domain. 

At $\mathtt{Rx}$, SI leaks to $\mathtt{Rx}$ through the time domain multipath channel, $h^\text{SI}$. The direct path between $\mathtt{Tx}$ and $\mathtt{Rx}$ generates SI as the direct coupling. Based on SI channel including multipath component, $h^\text{SI}\in\mathbb{C}^{1\times L}$, with the length of $L$, the received CP-added SI signal, $y^\text{SI,CP}_m[n]$ is as follows:
\begin{equation}
\begin{aligned}
y_m^\text{SI,CP}[n] &= h^\text{SI}[n]*x_m^\mathtt{Tx}[n] +z[n]\\
&= h^\text{SI}[n]*\left(f_\text{PA}(x_m^\text{CP}[n])\right)+z[n].
\end{aligned}
\label{eq.ysi}
\end{equation}
In the frequency domain, $\mathcal{Y}_m^\text{SI}$ based on $\mathcal{H}^\text{SI}=\mathtt{DFT}(h^\text{SI})$ is as follows: 
\begin{equation}
\mathcal{Y}_m^\text{SI}[p]=\mathcal{H}^\text{SI}[p]\tilde{\mathcal{X}}_m^\mathtt{Tx}[p],
\label{eq.runninglin}
\end{equation}
where $\tilde{\mathcal{X}}_m^\mathtt{Tx}$ is effective transmitted baseband signal. With DFT and IDFT, $\tilde{\mathcal{X}}_m^\mathtt{Tx}$ is as follows:
\begin{equation}
\begin{aligned}
\tilde{\mathcal{X}}_m^\mathtt{Tx}=\mathtt{DFT}\left\{f_\text{PA}(x^\text{IQ}_m[n])\right\}
=\mathtt{DFT}\left\{f_\text{PA}\left(\mathtt{IDFT}\left\{\mathcal{X}^\text{IQ}_m\right\}\right)\right\}.
\label{eq.fxtx}
\end{aligned}
\end{equation}

\subsection{Nonlinear and Parallel-Hammerstein Model-based SIC}
From (\ref{eq.runninglin}), the estimation of $\hat{\mathcal{Y}}_\text{SI}$ in digital SIC is as follows:
\begin{equation}
\mathcal{Y}_m^\text{UL}[p]=\mathcal{Y}_m^\mathtt{Rx}[p]-\hat{\mathcal{H}}[p]\tilde{\mathcal{X}}_m^\mathtt{Tx}[p],
\label{eq.SIClin}
\end{equation}
where $\hat{\mathcal{H}}^\text{SI}[p]$ is estimated SI channel. However, $\hat{\mathcal{H}}^\text{SI}$ cannot reflect the nonlinearity of SI signal. Nonlinear digital SIC, therefore, needs the modification of $\tilde{\mathcal{X}^\mathtt{Tx}_m}$ considering $x_m$ from (\ref{eq.xtx}).

A nonlinear PA generates IMD phenomena over frequency domain. The recent works on digital SIC assumes PA function $f_\text{PA}$ as polynomial as follows:
\begin{equation}
\begin{aligned}
f_\text{PA}(x_m^\text{IQ}[n])= \sum_{k=0}^{K_\text{max}} {a_{2k+1}\phi_{2k+1,m}[n]},
\label{eq.f_pa}
\end{aligned}
\end{equation}
where $\phi_{2k+1}[n]$ is the nonlinear basis with $(2k+1)$ order, $a_{2k+1}$ is the nonlinear coefficient for each basis, and $K_\text{max}$ is the maximum nonlinear order allowed by the BS. For simplicity, we represent PA coefficient vector, $\mathbf{a}$, as $\mathbf{a}=\left[a_1,a_3,\cdots,a_{2K_\text{max}+1}\right]^\top$. Parallel-Hammerstein (PH) model assumes the basis $\phi_{2k+1,m}$ as follows:
\begin{equation}
\begin{aligned}
\phi_{2k+1,m}[n] = x_m^\text{IQ}[n]\left|x^\text{IQ}_m[n]\right|^{2k}=\left(x^\text{IQ}_m[n]\right)^{k+1}\left((x^\text{IQ}_m[n])^*\right)^k.
\end{aligned}
\end{equation}
As the variation of PH model, some works assumes $\phi_{2k+1}$ as linear combination of  $\phi_0,\cdots,\phi_{2k-1}$~\cite{R_ortho, K_theo}.

The nonlinear basis to address IQ imbalance alongside PA can be defined by extending $\phi_{2k+1}$ to $\tilde{\phi}_{k_1,k_2}$ as follows:
\begin{equation}
\begin{aligned}
\tilde{\phi}_{k_1,k_2,m}[n]=x_m[n]^{k_1}(x_m[n]^*)^{k_2}.
\end{aligned}
\end{equation}
In this expanded model, the equation $\phi_{2k+1}=\tilde{\phi}_{k+1,k}$ holds. With conventional PH model, only one nonlinear basis sequence, $\phi_{2k+1}$, is needed for each $k$. This implies that for addressing $K_\text{max}$ levels of nonlinearity, the sequence is calculated $K_\text{max}$ times. With the introduction of $\tilde{\phi}_{k_1,k_2}$ as the nonlinear basis~\cite{K_basis}, for each $k$, sequences for all $(k_1,k_2)$ pairs satisfying $2k+1=k_1+k_2$ must be computed. Consequently, $2k$ sequence calculations are required, necessitating a total of $K_\text{max}(K_\text{max}+1)$ calculations. This approach significantly raises the complexity involved in computing the nonlinear basis.

\subsubsection{Frequency domain SIC}

We introduced a conventional time-domain digital SIC method using a circulant matrix of $|x_m^\text{IQ}|^{2k}x_m^\text{IQ}$, as described in \cite{SBFD_hwi}, where IQ imbalance is already conmpensated. On the other hand, the conventional frequency domain SIC estimates SI coefficients, $\hat{\pmb{\mathcal{C}}}_p\in\mathbb{C}^{1\times|\mathbb{K}p|}$, for each subcarrier $p$.
For each subcarrier $p$, $\mathbb{K}_p$ represents the set of selected nonlinear bases $\mathit{\phi}_{2k+1}$ or $\tilde{\phi}_{k_1,k_2}$.
The upper bound of $|\mathbb{K}_p|$ is $K_\text{max}$ when using $\phi_{2k+1}$, and $K_\text{max}(K_\text{max}+1)$ when using $\tilde{\phi}_{k_1,k_2}$.

From (\ref{eq.fxtx}), the nonlinear basis in frequency domain, $\mathit{\Phi}_{2k+1,m}$, is as follows:
\begin{equation}
\begin{aligned}
\mathit{\Phi}_{2k+1,m}=\mathtt{DFT}(\phi_{2k+1,m})
=\mathtt{DFT}\left(\left|x^\text{IQ}_m\right|^{2k}x^\text{IQ}_m\right).
\label{eq.genP}
\end{aligned}
\end{equation}
The frequency domain SIC estimates $\hat{\pmb{\mathcal{C}}}_p\in\mathbb{C}^{(K_\text{max}+1)\times 1}$ as follows:
\begin{equation}
\begin{aligned}
\hat{\pmb{\mathcal{C}}}_p=\text{argmin}\left(\pmb{\mathcal{Y}}^\mathtt{train}_p-\pmb{\mathit{\Phi}}^\mathtt{train}_p\hat{\pmb{\mathcal{C}}}_p\right),
\label{eq.estiC}
\end{aligned}
\end{equation}
where $\pmb{\mathcal{Y}}^\mathtt{train}_p\in\mathbb{C}^{|\mathbb{M}^\mathtt{train}|\times 1}$ and $\pmb{\mathit{\Phi}}^\mathtt{train}_p\in\mathbb{C}^{|\mathbb{M}^\mathtt{train}|\times |\mathbb{K}_p|}$ are as following:
\begin{equation}
\begin{rcases}
&\pmb{\mathcal{Y}}^\mathtt{train}_p=[\cdots,\mathcal{Y}_m[p],\cdots]^\top\\
&\pmb{\mathit{\Phi}}^\mathtt{train}_p=\left[\cdots,\pmb{\mathit{\Phi}}_{m}[p],\cdots\right]^\top
\end{rcases}
\forall m\in\mathbb{M}^\mathtt{train}.
\end{equation}
For each subcarriers, the vector of basis for symbol $m$, $\pmb{\mathit{\Phi}}_{m}[p]\in\mathbb{C}^{|\mathbb{K}_p|\times 1}$, is as follows:
\begin{equation}
\pmb{\mathit{\Phi}}_{m}[p]=\left[\mathcal{X}_m[p]\cdots\mathit{\Phi}_{2k+1,m}[p]\cdots\right]^\top\quad\forall k\in\mathbb{K}_p.
\end{equation}
In running SIC step, however, the system utilizes the transmit signal for every single symbol as $\pmb{\mathit{\Phi}}^\mathtt{train}_p\rightarrow\pmb{\mathit\Phi}_m[p]$. We define the set of $m$ that in running SIC step as $\mathbb{M}^\mathtt{run}$. From (\ref{eq.runninglin}), regeneration of $\hat{\mathcal{Y}}_m^\text{SI}$ is as follows:
\begin{equation}
\hat{\mathcal{Y}}^\text{SI}_m[p]=\pmb{\mathit\Phi}_m[p]\hat{\pmb{\mathcal{C}}}_p\quad\forall m\in\mathbb{M}^\mathtt{run}.
\end{equation}
When $x^\text{DL}_m[n]$ is given, the complexity of calculating each $k$th nonlinear basis, $\mathit{\Phi}_{2k+1,m}$, is as follows:
\begin{equation}
\mathit{\Phi}_{2k+1,m}=\mathtt{DFT}\left(|x^\text{IQ}_m[n]|^{2k}x^\text{IQ}_m[n]\right)\rightarrow\mathcal{O}\left(\frac{1}{2}P\log_2P+P\right). 
\end{equation}
Computing $\phi_{2k+1,m}=|x^\text{IQ}_m[n]|^{2k}x^\text{IQ}_m[n]$ results in $\mathcal{O}(P)$ complexity as it involves element-wise multiplication of $|x^\text{IQ}_m[n]|^2$ and $\phi_{2k-1,m}$. Computing $\mathit{\Phi}_{2k+1,m}=\mathtt{DFT}(\phi_{2k+1,m})$ has a complexity of $\mathcal{O}(\frac{1}{2}P\log_2P)$. 

\begin{figure*}
\begin{center}
    \subfigure[Analog beamformed SI channel.]{%
      \includegraphics[height=0.45\columnwidth,keepaspectratio]
      {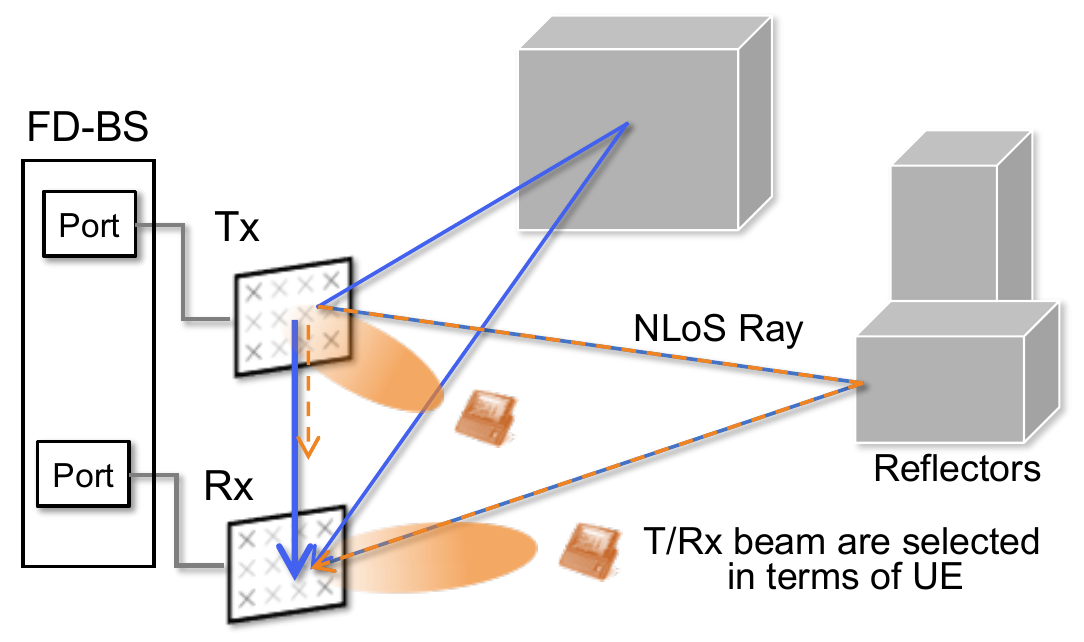}%
      \label{subfig.NLOS}
    }
    \subfigure[Time domain SI channel.]{%
      \includegraphics[height=0.45\columnwidth,keepaspectratio]
      {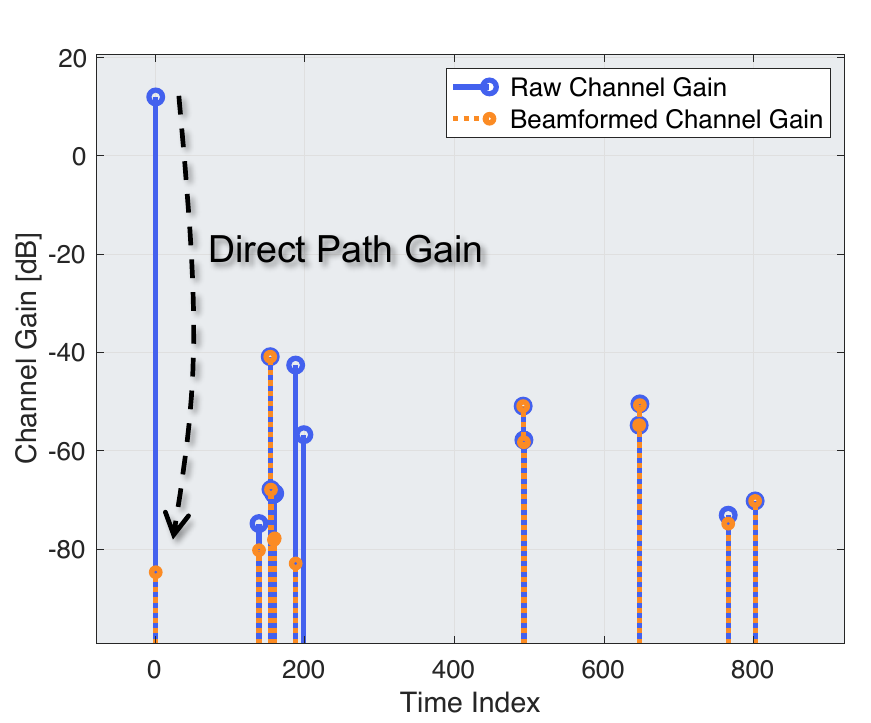}%
      \label{subfig.TimeD}
    }    
    \subfigure[Frequency domain SI channel.]{%
      \includegraphics[height=0.45\columnwidth,keepaspectratio]
      {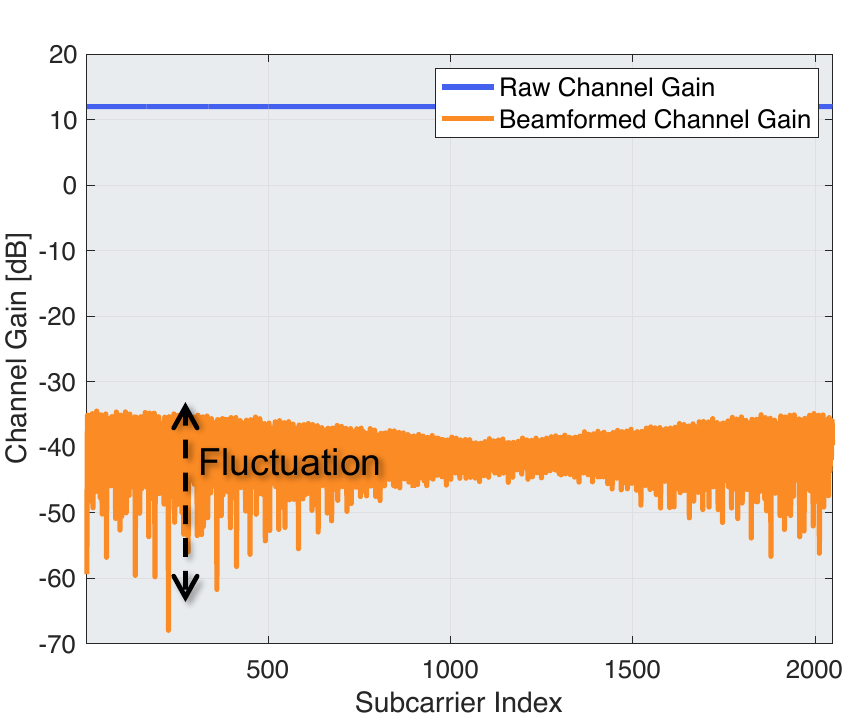}%
      \label{subfig.FreqD}
    }
 \end{center}
	\caption{Illustration of frequency-selectivity of multipath SI channel, based on analog beamforming.}
	\label{fig.FreqSel}
\vspace{-15pt}
\end{figure*}

\subsection{Intermodulation Distortion in Frequency Domain}

To analyze the SIC method and IMD phenomena in the frequency domain, we can rewrite each basis, $\phi_{2k+1}[n]$, based on OFDM as follows:
\begin{equation}
\begin{aligned}
&\phi_{2k+1,m}[n]=\left(x^\text{IQ}_m[n]\right)^{k+1}\left(x^\text{IQ}_m[n]^*\right)^{k}\\
&= \left(\frac{1}{P}\sum \mathcal{X}^\text{IQ}_m[p]e^{jw_pn}\right)^{k+1}\left\{\left(\frac{1}{P}\sum \mathcal{X}^\text{IQ}_m[p]e^{jw_pn}\right)^*\right\}^k\\
&=\frac{1}{P^{2k+1}}\left(\sum \mathcal{X}^\text{IQ}_m[p]e^{jw_pn}\right)^{k+1}\left(\sum \left(\mathcal{X}^\text{IQ}_m[p]e^{-jw_pn}\right)^*\right)^{k}\\
&=  \frac{1}{P^{2k+1}}\sum_{q_1,\cdots,q_{2k+1}\in\mathbb{P}}\left(\prod_{i=1}^{k+1} \mathcal{X}^\text{IQ}_m[q_i]\prod_{i=k+2}^{2k+1}\left(\mathcal{X}^\text{IQ}_m[q_i]\right)^*\right)e^{jwn},
\end{aligned}
\end{equation}
where $w_p$ and $w$ are as follows:
\begin{equation}
w_p =  \frac{2\pi p}{P},\quad w = \frac{2\pi}{P}\left(\sum_{i=1}^{k+1}q_i- \sum_{i=k+2}^{2k+1}{q_i}\right).
\end{equation}
For SIC, we only have to obtain the nonlinear basis of frequency range of UL. Applying DFT operation to $\phi_{2k+1}$, we can get nonlinear basis of UL band, $\mathit\Phi_{2k+1}=\mathtt{DFT}\left(\phi_{2k+1}\right)$, as follows:
\begin{equation}
\begin{aligned}
&\mathit\Phi_{2k+1,m}[p]=(1/{P^{2k+1}})\\
&\times\sum_{(q_1,\cdots,q_{2k+1})\in\mathbb{Q}^{2k+1}_p}\left(\prod_{i=1}^{k+1} \mathcal{X}^\text{IQ}_m[q_i]\prod_{i=k+2}^{2k+1}\left(\mathcal{X}^\text{IQ}_m[q_i]\right)^*\right),
\label{eq.PhiUL}
\end{aligned}
\end{equation}
where $\mathbb{Q}^{2k+1}_p$ is the IMD set of sequence $(q_1,q_2\cdots,q_{2k+1})$ for each $p\in\mathbb{P}^\text{UL}$. We can obtain the condition of $\mathbb{Q}^{2k+1}_p$ in terms of $p$ as follows: 
\begin{equation}
\begin{aligned}
\mathbb{Q}^{2k+1}_p=&\left\{(q_1,q_2,\cdots,q_{2k+1})|\sum_{i=1}^{k+1}q_i-\sum_{i=k+2}^{2k+1}q_i=\tilde{p}\right\},\\
\tilde{p} =& p+r P \quad \forall r\in \mathbb{R},
\label{eq.setQ}
\end{aligned}
\end{equation}
where $r$ is the arbitrary real number. As $\tilde{p}$ has periodicity, the overall procedure to obtain $\mathit{\Phi}_{2k+1,m}$ needs to consider the same characteristic.
For example, the 3rd nonlinear basis is as follows:
\begin{equation}
\begin{aligned}
\mathit\Phi_{3,m}[p]=\frac{1}{P^3}\sum_{(q_1,q_2,q_3)\in\mathbb{Q}^{3}_p}\left(\mathcal{X}^\text{IQ}_m[q_1]\mathcal{X}^\text{IQ}_m[q_2]\left(\mathcal{X}^\text{IQ}_m[q_3]\right)^*\right).
\label{eq.IMDbase}
\end{aligned}
\end{equation}

\subsection{MIMO SI Channel Modeling \& Analog Beamforming}
 {The system under study incorporates the analog beamforming. Analog beamforming is implemented at both $\mathtt{Tx}$ and $\mathtt{Rx}$, across $N_\text{AE}$ antenna elements (AE) at each location. }

\subsubsection{Multipath MIMO SI channel}

 {To model the MIMO SI channel, we approached each ray originating from Tx by allowing the number of reflections and diffractions to vary, with the total number of rays, $N_\text{ray}$, determined based on which rays ultimately reach Rx. Since each ray arrives at Rx at different times, the multipath channel exhibits a multi-tap structure with frequency-selective characteristics. To accurately capture this behavior in a MIMO channel, we modeled and generated a multipath MIMO channel by calculating the array response based on the number of multipaths between Tx and Rx, along with the AoA and AoD of each ray.}

 {In further modeling the SI channel within this multipath environment, we define each ray by its index $\ell$, where $N_\text{ray}$ represents the total number of rays. For $\ell=1$, we consider the line-of-sight (LoS) scenario, while for $\ell\geq2$, we account for non-line-of-sight (NLoS) conditions. The time index for the LoS ray is denoted as $n^\text{LoS}$. Based on this setup, we define the multipath MIMO SI channel in the time domain, $\mathbf{H}^\text{SI}[n]\in\mathbb{C}^{N_\text{AE}\times N_\text{AE}}$, as follows:
\begin{equation}
\mathbf{H}^\text{SI}[n]=\sum_{\ell=1}^{N_\text{ray}}g_\ell \delta^\mathtt{imp}(nT_\text{sampling}-\tau_\ell)\mathbf{e}_\mathtt{Rx}(\theta_\text{AoA})\mathbf{e}_\mathtt{Tx}^\mathsf{H}(\theta_\text{AoD}),
\label{eq.MIMOa}
\end{equation}
where $g_\ell$ and $\tau_\ell$ are channel gain and delay time of the $\ell$th ray, respectively. And $T_\text{sampling}=1/(N\Delta f)$ is sampling interval, function of subcarrier spacing, $\delta^\mathtt{imp}$ is the dirac-delta impluse function, and $\mathbf{e}_\mathtt{Rx}(\theta_\text{AoA})$, $\mathbf{e}_\mathtt{Tx}(\theta_\text{AoD})$ are array response.}

 {For the channel model, models such as 3GPP TR~38.901 can be applied, but particularly in the case of $\ell=1$, where the close distance between Tx and Rx results in a near-field LoS SI channel, as reflected in the above model. For the primary LoS ray expected to arrive first, a MIMO channel was created using the array response, moving away from the far-field wavefront assumption~\cite{IanBSI,IanMag}.}

\subsubsection{Analog beamforming system}
The analog beamforming functions as a steering vector based on the phase shifter at the BS. The analog beams are represented as $\mathbf{f}\in\mathbb{C}^{N_\text{AE}\times 1}$ for $\mathtt{Tx}$ and $\mathbf{w}\in\mathbb{C}^{N_\text{AE}\times 1}$ for $\mathtt{Rx}$, and both satisfy the unit modulus constraint. Optimization of DL channel gain involves pairing $\mathtt{Tx}$ of the BS and $\mathtt{Rx}$ of the DL-UE with the analog beams, denoted as $\mathbf{f}^\text{DL}_\text{BS}$ and $\mathbf{w}^\text{DL}_\text{UE}$. For UL, the analog beam alignment process involves $\mathbf{w}^\text{UL}_\text{BS}$ to $\mathtt{Rx}$ of the BS, and from $\mathtt{Rx}$ of the UL-UE, represented by $\mathbf{f}^\text{UL}_\text{UE}$.

\subsubsection{Analog beamforming and freqeuncy-selectivity of SI}
The SI channel, $\mathbf{H}^\text{SI}[n]\in\mathbb{C}^{N_\text{AE}\times N_\text{AE}}$ in time domain, and $\pmb{\mathcal{H}}^\text{SI}[p]\in\mathbb{C}^{N_\text{AE}\times N_\text{AE}}$ in the frequency domain, means the signal transmitted from AEs of the BS $\mathtt{Tx}$ returns to AEs of the BS $\mathtt{Rx}$. When $N_\text{RF}=1$, The effective SI channel in the time domain, $h^\text{SI}[n]$, is as follows:
\begin{equation}
h^\text{SI}[n]=\left(\mathbf{w}^\text{UL}_\text{BS}\right)^\mathsf{H}\mathbf{H}^\text{SI}[n]\mathbf{f}^\text{DL}_\text{BS}.
\end{equation}
The effective SI channel in the frequency domain, $\mathcal{H}^\text{SI}[p]$, is as follows:
\begin{equation}
\mathcal{H}^\text{SI}[p]=\left(\mathbf{w}^\text{UL}_\text{BS}\right)^\mathsf{H}\pmb{\mathcal{H}}^\text{SI}[p]\mathbf{f}^\text{DL}_\text{BS}.
\end{equation}

The impact of the analog beamforming on the SI channel within a multipath environment can be observed as depicted in Fig.\ref{fig.FreqSel},. Determining the analog beam is unique to each DL-UE and UL-UE, leading to frequency-selectivity naturally associated with the analog beam.

Fig.\ref{subfig.NLOS} illustrates the coexistence of LoS and NLoS paths due to multipath in the SI channel. Fig.\ref{subfig.TimeD} provides a plot of $h^\text{SI}[n]$ over time step $n$, showing that the direct path gain tends to be the most significant. The direct path gain decreases in response to the analog beam while the NLoS gain increases. The decrease in the LoS gain results in frequency-selectivity within the SI channel in the frequency domain, $\mathcal{H}^\text{SI}[p]$, as shown in Fig.\ref{subfig.FreqD}.

\section{Problem Formulation}

The main objective of this study is to develop a low-complexity SIC method suitable for flexible duplex systems. While recent research on low-complexity SIC has primarily focused on IBFD, there are opportunities to adapt such methods for flexible duplex scenarios. From the perspective of frequency domain SIC, the SIC process can be divided into two main steps: the \emph{estimation step} and the \emph{running step}. To achieve low complexity, it is essential to minimize the operations involved in generating $\pmb{\mathit{\Phi}}_m[p]$ during the running step and to reduce the complexity of estimating $\hat{\pmb{\mathcal{C}}}_p$ in the estimation step.

\subsubsection{Estimation step} estimation of $\hat{\pmb{\mathcal{C}}_p}$ when $m\in\mathbb{M}^\mathtt{train}$. The system compose $\pmb{\mathit{\Phi}}_p^\mathtt{train}$ from given $\mathcal{X}_{m}$. One of the multi-variable regression method, least square (LS) method estimates $\pmb{\mathcal{C}}_p$ to satisfies (\ref{eq.estiC}) as follows:
\begin{equation}
\begin{aligned}
&\hat{\pmb{\mathcal{C}}}_p=\left((\pmb{\mathit{\Phi}}^\mathtt{train}_p)^\mathsf{H}\pmb{\mathit{\Phi}}^\mathtt{train}_p\right)^{-1}(\pmb{\mathit{\Phi}}^\mathtt{train}_p)^\mathsf{H}\pmb{\mathcal{Y}}^\mathtt{train}_p,\\
&\rightarrow\mathcal{O}\left(\left(\max_{p\in\mathbb{P}^\text{UL}}|\mathbb{K}_p|\right)\left(\frac{1}{2}P\log_2P+P\right)+\sum_{p\in\mathbb{P}^\text{UL}}|\mathbb{K}_p|^3\right).
\end{aligned}
\label{eq.complex1}
\end{equation}
As represented in (\ref{eq.complex1}), the complexity per symbol is primarily denoted by $\mathcal{O}\left(\sum|\mathbb{K}_p|^3\right)$. To inherently reduce complexity, the size of $\pmb{\mathit{\Phi}}_p^\mathtt{train}$ can be diminished by judiciously selecting the basis function $\mathit{\Phi}_{2k+1,m}[p]$ from $\mathbb{K}_p$.

\subsubsection{Running step} regeneration of $\hat{\mathcal{Y}}^\text{SI} = \pmb{\mathit{\Phi}}_{m}[p]\hat{\pmb{\mathcal{C}}}_p$ when $m\in~\mathbb{M}^\mathtt{run}$. Given $\hat{\pmb{\mathcal{C}}}_p$, the system needs to generate $\pmb{\mathit{\Phi}}_{m}[p]$ from $\mathcal{X}_m$ at every symbol as in (\ref{eq.genP}). 
Including the mapping of the nonlinear basis, the complexity of running step is as follows:
\begin{equation}
\begin{aligned}
&\hat{\mathcal{Y}}^\text{SI} = \pmb{\mathit{\Phi}}_{m}[p]\hat{\pmb{\mathcal{C}}}_p,\\
&\rightarrow\mathcal{O}\left(\left(\max_{p\in\mathbb{P}^\text{UL}}|\mathbb{K}_p|\right)\left(\frac{1}{2}P\log_2P+P\right)+\sum_{p\in\mathbb{P}^\text{UL}}|\mathbb{K}_p|\right).
\label{eq.complex2}
\end{aligned}
\end{equation}
Overall, effectively defining and managing $\mathbb{K}_p$ is crucial for minimizing complexity in both the estimation and the running step.


\subsection{IMD Nonlinear Basis in Flexible Duplex Scenario}

The complexity depends on the number of $\mathit{\Phi}_{2k+1,m}[p]$ elements included in $\pmb{\mathit{\Phi}}_p^\mathtt{train}$.
The power of the IMD is determined by the allocation of the DL band, denoted as $\mathbb{P}^\text{DL}$.
An analysis of the IMD generated by the nonlinear basis, based on the general $\mathbb{P}^\text{DL}$ of flexible duplex, can modify $\pmb{\mathit{\Phi}}_{m}[p]$, as well as $\mathbb{K}_p$, thereby reducing the complexity of digital SIC.

Previous studies have proposed some low-complexity strategies, such as basis modification. One approach involved orthogonalizing $\mathit{\Phi}_{2k+1,m}$ in the frequency domain and estimating the SI coefficient on the adjusted basis~\cite{FDC_iter}. Another method determined the inclusion of $\mathit{\Phi}_{2k+1,m}$ in $\pmb{\mathit{\Phi}}_m[p]$ for SI coefficient estimation based on its influence~\cite{K_basis}.
The SI power due to the $2k+1$th nonlinear basis, denoted $\hat{\mathcal{I}}_{2k+1}[p]$ is as follows:
\begin{equation}
\hat{\mathcal{I}}_{2k+1}[p]=\hat{a}_{2k+1}^2\mu_{2k+1}[p]\left|\hat{\mathcal{H}}^\text{SI}[p]\right|^2.
\label{eq.Ebase}
\end{equation}
The inclusion of $\mathit{\Phi}_{2k+1,m}$ in $\pmb{\mathit{\Phi}}_m[p]$ is determined as follows: 
\begin{equation}
\begin{aligned}
\hat{\mathcal{I}}_{2k+1}[p]>\Gamma 
\rightarrow k\in\mathbb{K}_p,
\end{aligned}
\label{eq.inclusion}
\end{equation}
where $\Gamma$ is threshold based on the noise level, $E\left[|\mathcal{Z}[p]|^2\right]$, and 
$\mathbb{K}_p$ represents the set of $\mathit{\Phi}_{2k+1}$ included in the estimation and running SIC process.

However, both \cite{K_basis}, \cite{FDC_iter} have yet to derive a closed-form solution for $E\left[|\mathit{\Phi}_{2k+1,m}[p]|^2\right]$, necessitating additional operations and signalling. We investigate our novel idea for the low-complexity SIC as follows. The effective IMD nonlinear basis, $\mathit{\Phi}_{2k+1}$ from (\ref{eq.PhiUL}), necessitates frequency domain analysis for adapting to various flexible duplex situations. The closed-form result of IMD nonlinear basis reduces complexity by adaptively determining the basis for each subcarrier. 

Additionally, to utilize this method, prior knowledge of $\hat{a}_{2k+1}$ is necessary. The following subsection will detail the process for estimating the PA coefficient, which remains constant irrespective of the fluctuating channel, in advance through the LoS SI channel.

\subsection{Leveraging Beamformed SI Channel to Decomposing the Nonlinearity}
Initial approaches combine the SI channel $\mathcal{H}^\text{SI}$ and $\mathbf{a}$ to estimate $\hat{\pmb{\mathcal{C}}}_p=\hat{\mathbf{a}}\hat{\mathcal{H}}^\text{SI}[p]$. Separating the SI channel and the $\mathtt{T/Rx}$ nonlinearity offers another approach. This two-step process estimates $\mathcal{H}^\text{SI}$, subsequently evaluate $\mathbf{a}$ as follows:
\begin{equation}
\hat{\mathbf{a}}=\text{argmin}\left(\mathcal{Y}[p]-\left(\mathcal{H}^\text{SI}[p]\pmb{\mathit{\Phi}}_p\right)\hat{\mathbf{a}}\right).
\label{eq.estia}
\end{equation}
An implementation strategy is to iteratively calculate $\mathcal{H}^\text{SI}[p]$, and $\mathbf{a}$ in a rotating manner~\cite{K_iter}. This method necessitates guard band information when estimating $\mathbf{a}$, thus demanding supplemental measurement. 

\begin{figure*}
\begin{center}
    \subfigure[Theorem 2 for $(2k+1)$th IMD basis.]{%
      \includegraphics[width=0.69\columnwidth,keepaspectratio]
      {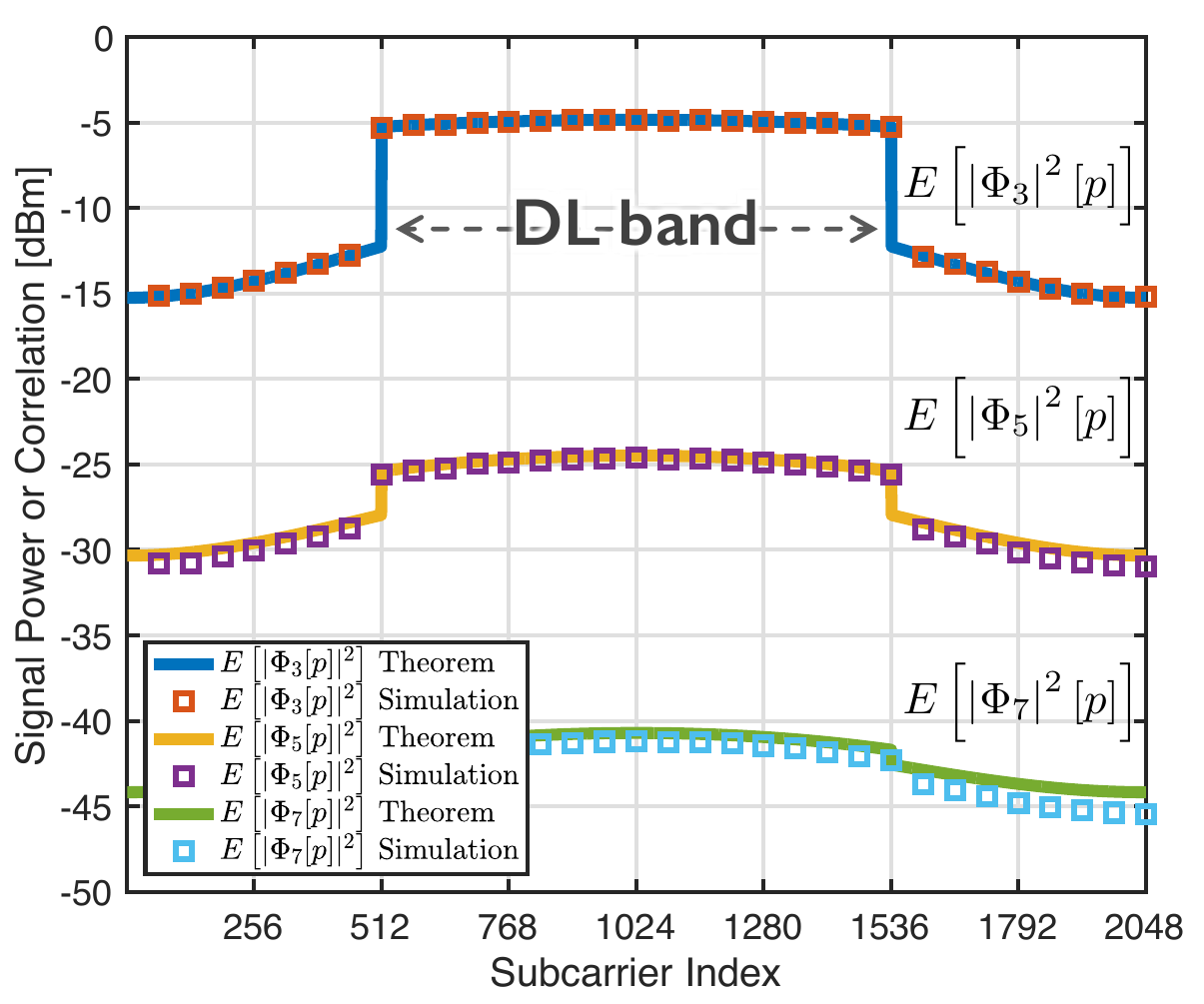}%
      \label{subfig.Esub}
    }
    \subfigure[Proposition 2.1 in terms of various DL band.]{%
      \includegraphics[width=0.69\columnwidth,keepaspectratio]
      {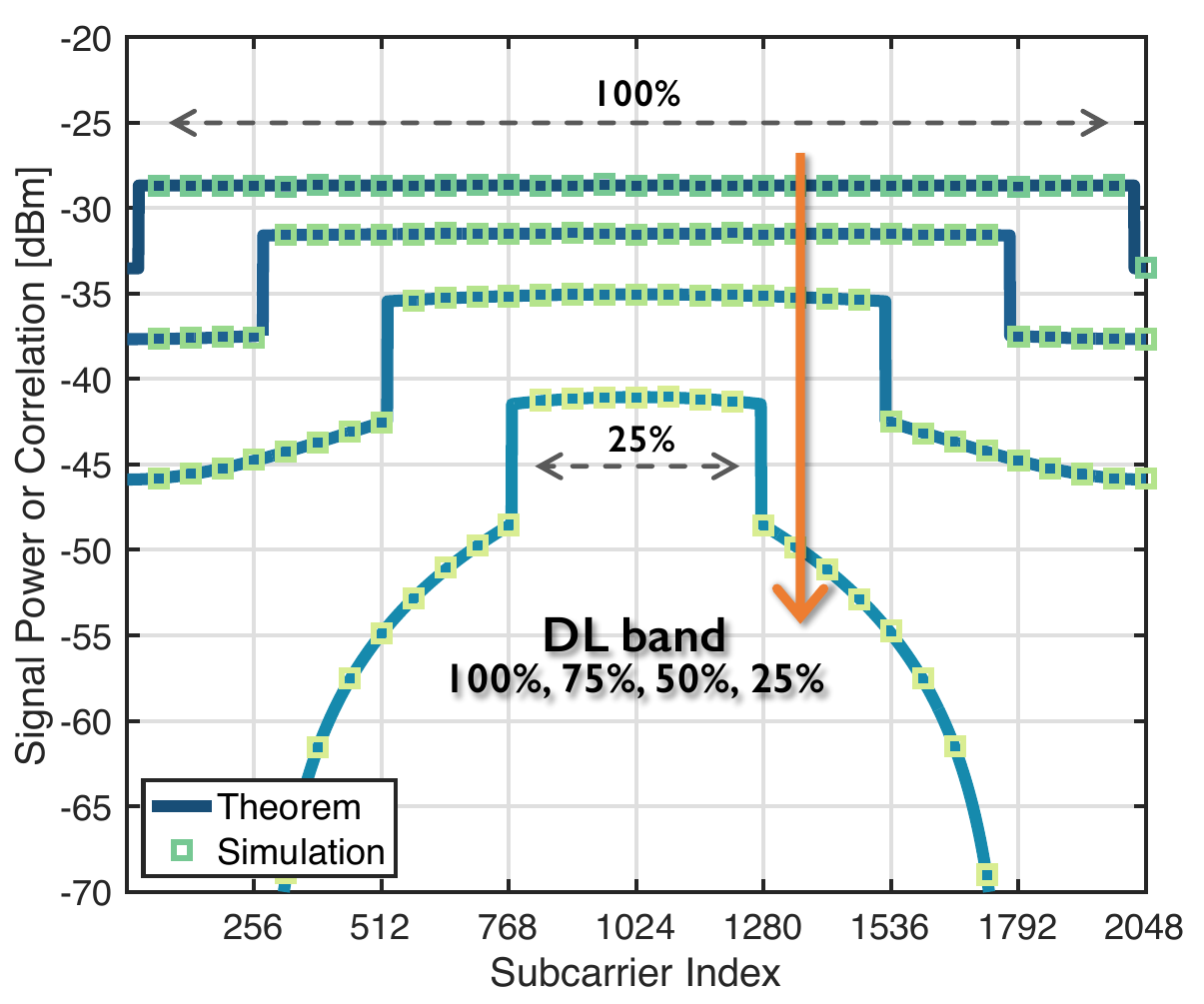}%
      \label{subfig.Ssub}
    }    
    \subfigure[Summation for multiple subcarriers.]{%
      \includegraphics[width=0.60\columnwidth,keepaspectratio]
      {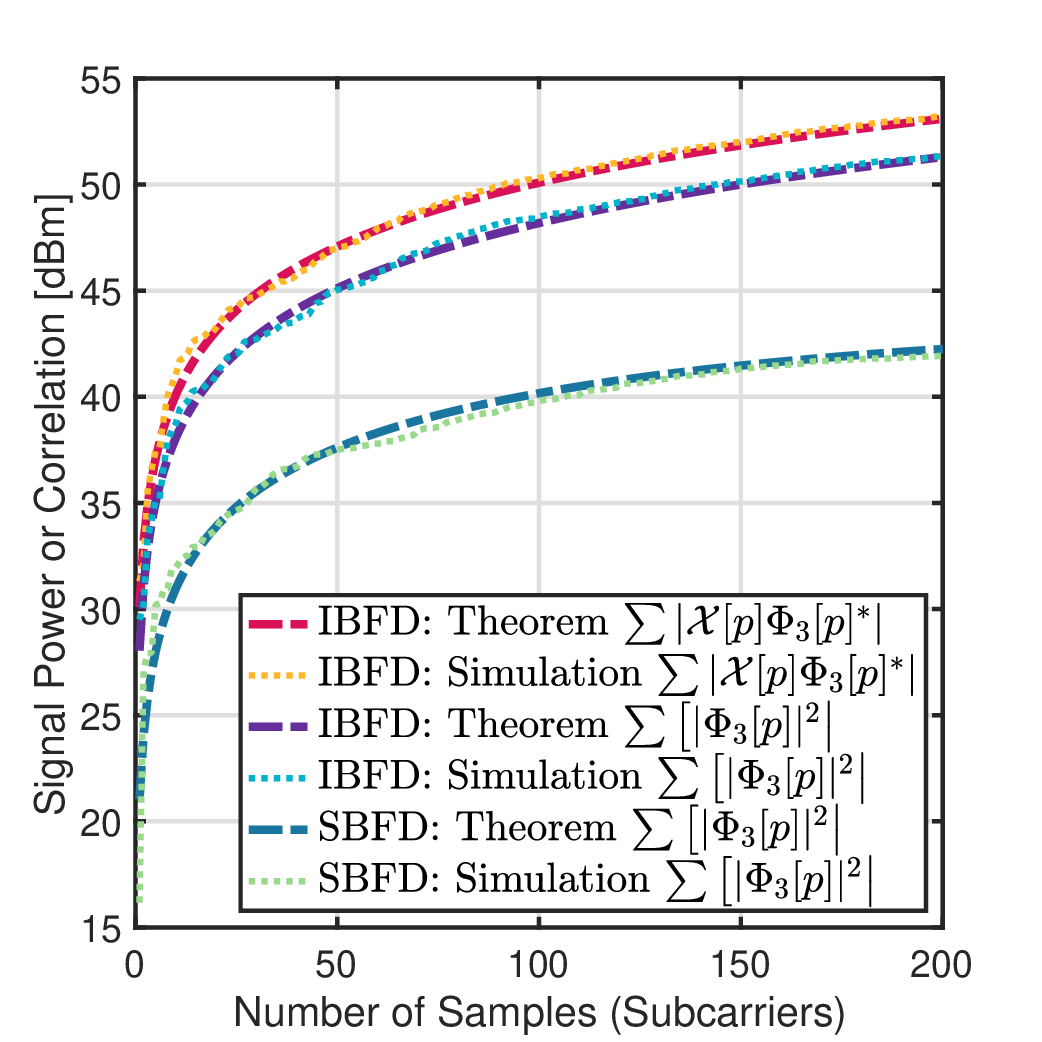}%
      \label{subfig.Ssub}
    }
 \end{center}
	\caption{The results show the exactness of the proposed nonlinear SI analysis in the flexible duplex. A comparison about $3$rd, $5$th, and $7$th nonlinear basis was conducted between Theorem 2, and numerical results. Furthermore, an examination between closed-form of Proposition 2.1 and numerical results was conducted, in terms of the size of the DL band for the $3$rd IMD basis.}
	\label{fig.TheoVal}
\vspace{-15pt}
\end{figure*}

The digital SIC estimation step relies on the SI generated by the DL signal $\mathcal{X}_m$ at $\mathtt{Tx}$, transformed to $\pmb{\mathit{\Phi}}_p^\mathtt{train}$, and received as $\mathcal{Y}^\mathtt{Rx}=\mathcal{Y}^\text{SI}$ at $\mathtt{Rx}$, when only a DL signal is present without a UL signal. In this scenario, generating a pilot signal capable of producing the nonlinear basis $\mathit{\Phi}_{2k+1}$ efficiently, without sacrificing estimation accuracy, is crucial. A new process enables the separate estimation of the SI coefficient's components: $\mathcal{H}^\text{SI}$, $a_{2k+1}$, and $b^\text{IQ}$. Traditionally, estimating all three components requires a complexity of $\mathcal{O}(P^\text{UL}(K_\text{max}(K_\text{max}+1))^3)$ with the given $\pmb{\mathit{\Phi}}_p^\mathtt{train}$. However, by isolating $\pmb{\mathcal{C}}_p$ and $b^\text{IQ}$, the complexity can drop to $\mathcal{O}(P^\text{UL}(K_\text{max}+1)^3)$. If $a_{2k+1}$ is frequency-independent, further simplification to $\mathcal{O}(P^\text{UL}+|\mathbb{K}_p|)$ is achievable by distinguishing $\mathcal{H}^\text{SI}$ and $a_{2k+1}$. This separation allows the application of conventional OFDM channel estimation techniques, accommodating changes in the SI channel over time.

The analog beamforming or BFC can interpret the SI channel through the LoS component, unlike other analog SICs. 
We aims to design a pilot to estimate nonlinearity, $\mathbf{a}$ and $b^\text{IQ}$, based on the beamformed LoS SI channel, $\mathcal{H}^\text{LoS}$.
Estimating and distinguishing the $\mathbf{a}$ and $b^\text{IQ}$ from $\pmb{\mathcal{C}}_p$, therefore, can lower the complexity.

\section{Proposed Method for Nonlinear, \\low-complexity, frequency domain SIC}

In this section, we conduct a detailed analysis of nonlinear SI using the formulated theorems. Considering the flexible duplex, we evaluate SI from IMD in the frequency domain based on analysis of the nonlinear basis, $\mathit\Phi_{2k+1,m}$. We introduce a low-complexity SIC algorithm and a suitable pilot pattern to reduce complexity, as demonstrated in both (\ref{eq.complex1}) and (\ref{eq.complex2}).

\subsection{Nonlinear SI Analysis for Flexible Duplex: Reducing $|\mathbb{K}_p|$}

In flexible duplex scenarios with variable UL/DL subcarriers, our goal is to estimate the shape of SI using the expected value $\mu_{2k+1}=E\left[\left|\mathit\Phi_{2k+1,m}[p]\right|^2\right]$. Traditional methods view $E\left[\left|\mathit\Phi_{2k+1,m}[p]\right|^2\right]$ as an expectation over $m$, which requires the transmission and reception of many symbols for estimation. Using the following theorems and propositions, we offer a near-closed-form solution for $\mu_{2k+1}$, which changes according to $\mathbb{P}^\text{UL}$ and $\mathbb{P}^\text{DL}$. To support this, we first reveal the properties of $\mathit\Phi_{2k+1,m}[p]$. By utilizing these properties, we derive the characteristics of $\mathbb{Q}^{2k+1}_p$.

\begin{lem}
The effective $2k+1$th IMD nonlinear basis, $\mathit\Phi_{2k+1}[p]$, satisfies following recursive equation:
\begin{equation}
\begin{aligned}
&\mathit\Phi_{2k+1,m}[p]=\textup{IMD}\left(\mathcal{X}_m^\textup{IQ},\mathit{\Phi}_{2k+1,m}\right)\\
&=\sum_{q_1,q_2\in\mathbb{P}^\textup{DL}}\mathcal{X}^\textup{IQ}_m[q_1]\mathcal{X}^\textup{IQ}_m[q_2]\mathit\Phi^*_{2k-1,m}[q_1+q_2-p],
\label{eq.lem1.a}
\end{aligned}
\end{equation}
where $q_1,q_2\in\mathbb{P}^\textup{DL}$ are the indices of auxiliary subcarriers. With a similar concept, we can represent the alternative form of $\mathit\Phi_{2k+1}[p]$ as follows:
\begin{equation}
\begin{aligned}
&\mathit\Phi_{2k+1,m}[p]=\\
&\sum_{q_1,q_{k+2}\in\mathbb{P}^\textup{DL}}\mathcal{X}^\textup{IQ}_m[q_1]\left(\mathcal{X}^\textup{IQ}_m[q_{k+2}]\right)^*\mathit\Phi_{2k-1,m}[p-q_1+q_{k+2}].
\end{aligned}
\label{eq.lem1.b}
\end{equation}
\end{lem} 
\begin{proof}
See Appendix A.
\end{proof}


\begin{thm}
The size of the $(2k+1)$th IMD set for sequence $q_i$ sequence, $\left|\mathbb{Q}^{2k+1}_p\right|$, the number of $(q_1,\cdots,q_{2k+1})$ that satisfies $\sum_{i=1}^{k+1} q_i-~\sum_{j=1}^{2k+1} q_j = p$, to utilize the proposed SI pilot signal is as follows:
\begin{equation}
\begin{aligned}
\left|\mathbb{Q}^{2k+1}_p\right|
&=\sum_{\rho\in\mathbb{P}}\Lambda^\textup{DL}[p+\rho]\left|\mathbb{Q}^{2k-1}_\rho\right|,
\label{eq.theo1}
\end{aligned}
\end{equation}
where $\Lambda^\textup{DL}$ is the triangular function for $\mathbb{P}^\textup{DL}$ as follows:
\begin{equation}
\begin{aligned}
\Lambda^\textup{DL}[p]=
\begin{cases}
p-2p^\textup{DL}_\textup{start}
& 2p^\textup{DL}_\textup{start}<p<p^\textup{DL}_\textup{start}+p^\textup{DL}_\textup{end}\\
-p+2p^\textup{DL}_\textup{end}
& p^\textup{DL}_\textup{start}+p^\textup{DL}_\textup{end}<p<2p^\textup{DL}_\textup{end}\\
0 & \textup{otherwise}
\end{cases}.
\label{eq.tri}
\end{aligned}
\end{equation}
\label{thm.NQ}
\end{thm}
\begin{proof}
From (\ref{eq.lem1.a}), we define the index, $\rho=q_1+q_2-p$. The number of pair, $(q_1,q_2)$, satisfying $q_1+q_2=p+\rho$, is triangular function in terms of $\mathbb{P}^\text{DL}$ as in (\ref{eq.tri}). By multiplying the, $\left|\mathbb{Q}_\rho^{2k-1}\right|$ for each $\Lambda^\text{DL}[p+\rho]$, the theorem follows.
\end{proof}

\begin{table*}[]
\caption{Complexity Analysis of the frequency domain, Nonlinear, SIC algorithm}
\centering
\resizebox{\textwidth}{!}{%
\renewcommand{\arraystretch}{1.4}
\begin{tabular}{c|c|cc|cc}
\toprule
\multirow{2}{*}{\textbf{Step}}                                                                                     & \multirow{2}{*}{\textbf{Description}}                                & \multicolumn{2}{c|}{\textbf{Proposed Method}}                                                                                                                                                                                                           & \multicolumn{1}{c|}{\multirow{2}{*}{\cite{K_basis}}} & \multirow{2}{*}{\cite{FDC_iter}} \\ \cline{3-4}
                                                                                                          &                                                             & \multicolumn{1}{c|}{PA Coefficient, $\mathbf{a}$}                                                                                          & Freq. SI Channel $\mathcal{H}^\text{SI}$                                                                                             & \multicolumn{1}{c|}{}                      &                       \\ \midrule
\multirow{5}{*}{\begin{tabular}[c]{@{}c@{}}Common\\ Step\\ to Obtain \\ $\mathit{\Phi}_{2k+1,m}$\end{tabular}} & {$\mathcal{X}_m^\text{IQ}[p]=\mathcal{X}_m[p]+b^\text{IQ}\mathcal{X}_m[-p]$} &  \multicolumn{4}{c}{$P^\text{DL}$}\\
\cline{2-6}
& $x^\text{IQ}_m=\mathtt{IDFT}(\mathcal{X}^\text{IQ}_m)$                        & \multicolumn{1}{c|}{\multirow{3}{*}{\begin{tabular}[c]{@{}c@{}}$0$\\ (Proposition 2.2)\end{tabular}}}                                & \multirow{3}{*}{\begin{tabular}[c]{@{}c@{}}$(P^\text{DL})^2/2$\\ $P+P\log_2P$\end{tabular}} & \multicolumn{2}{c}{$\frac{1}{2}P\log_2 P$}                         \\ \cline{2-2} \cline{5-6} 
                                                                                                          & $\phi_{2k+1,m}=|x_m^\text{IQ}|^{2k}x_m^\text{IQ}$                            & \multicolumn{1}{c|}{}                                                                                                          &                                                                                                               & \multicolumn{1}{c|}{$PK_\text{max}$}                  & $P$                   \\ \cline{2-2} \cline{5-6} 
                                                                                                          & $\mathit\Phi_{2k+1,m}=\mathtt{DFT}(\phi_{2k+1,m})$               & \multicolumn{1}{c|}{}                                                                                                          &                                                                                                               & \multicolumn{2}{c}{$\frac{1}{2}K_\text{max}P\log_2 P$}                        \\ \cline{2-6} 
                                                                                                          & \begin{tabular}[c]{@{}c@{}}Transform $\mathit\Phi_{2k+1,m}$\end{tabular} & \multicolumn{1}{c|}{\begin{tabular}[c]{@{}c@{}}$K_\text{max}+1$\\ \end{tabular}} & \begin{tabular}[c]{@{}c@{}}$\sum|\mathbb{K}_p|<K_\text{max}P$\\ \end{tabular}                        & \multicolumn{1}{c|}{-}                     & $P^\text{UL}(K_\text{max}+1)^2$  \\ \hline
\begin{tabular}[c]{@{}c@{}}Estimating \\ SI Coefficient\end{tabular}                                                                             & $\hat{C}[p]=\text{Estimator}({\pmb{\mathit\Phi}}_p,\mathcal{Y})$     & \multicolumn{1}{c|}{$(K_\text{max}+1)^3$}                                                                                                 & \begin{tabular}[c]{@{}c@{}}$P^\text{UL}$\\ \end{tabular}                            & \multicolumn{1}{c|}{$P^\text{UL}(K_\text{max}+1)^3$}  & $P^\text{UL}(K_\text{max}+1)R$   \\ \hline
Running SIC                                                                                               & $\hat{\mathcal{Y}}[p]={\pmb{\mathit\Phi}}_m[p]\hat{\pmb{\mathcal{C}}}_p$            & \multicolumn{2}{c|}{$P^\text{UL}+\sum|\mathbb{K}_p|<P^\text{UL}(K_\text{max}+1)$}                                                                                                                                                                         & \multicolumn{1}{c|}{$P^\text{UL}(K_\text{max}+1)$}    & $P^\text{UL}(K_\text{max}+1)R$   \\ \bottomrule
\end{tabular}%
}
\label{table.complexity}
\end{table*}

\begin{prop}
The size of the 3rd IMD set in terms of UL subcarrier $p\in\mathbb{P}^\textup{UL}$, $\left|\mathbb{Q}^3_p\right|$, is as follows:
\begin{equation}
\begin{aligned}
\left|\mathbb{Q}^{3}_p\right|=Q(p-P)+Q(p)+Q(p+P),
\end{aligned}
\end{equation}
where $Q(p)$ is as follows:
\begin{equation}
\begin{aligned}
Q(p) = \begin{cases}
   \frac{1}{2}\left(p+p^\text{DL}_\text{end}-2p^\text{DL}_\text{start}\right)^2 & 2p^\text{DL}_\text{start}-p^\text{DL}_\text{end} < p \leq p^\text{DL}_\text{start}\\
  Q_\text{IB}(p) & p^\text{DL}_\text{start}\leq p<p^\text{DL}_\text{end}\\
  \frac{1}{2}\left(2p^\text{DL}_\text{end}-p^\text{DL}_\text{start}-p\right)^2 & p^\text{DL}_\text{end} < p \leq 2p^\text{DL}_\text{end}-p^\text{DL}_\text{start}
\end{cases}\\
Q_\text{IB}(p)= \left(p^\text{DL}_\text{end}-p^\text{DL}_\text{start}\right)^2-\frac{1}{2}\left(p-p^\text{DL}_\text{start}\right)^2-\frac{1}{2}\left(p^\text{DL}_\text{end}-p\right)^2.
\label{eq.Q3p}
\end{aligned}
\end{equation}
\end{prop}
\begin{proof}
For the case of $k=1$, the condition $\mathbb{Q}^1_p=\mathbb{P}^\text{DL}$ holds. Equation (\ref{eq.theo1}) can be reshaped to $|\mathbb{Q}^3_p|=~\sum_{\rho\in\mathbb{P}^\text{DL}}\Lambda^\text{DL}[p+\rho]$. With reference to (\ref{eq.tri}), the summing of $\Lambda^\text{DL}[p+\rho]$ for certain $\rho$ as follows:
\begin{equation}
\max(p+p^\text{DL}_\text{start},2p^\text{DL}_\text{start})<p+\rho<\min(p+p^\text{DL}_\text{end},2p^\text{DL}_\text{end}).
\end{equation}
Aggregating the affective sections of $\Lambda^\text{DL}[p+\rho]$ according to equation (\ref{eq.tri}) results in equation (\ref{eq.Q3p}), and the proposition follows.
\end{proof}

\begin{thm}
The expectation of the $(2k+1)$th nonlinear basis over $m$, $\mu_{2k+1}[p]=E\left[\left|{\mathit\Phi_{2k+1,m}[p]}\right|^2\right]$, to predict nonlinear SI, $\hat{\mathcal{I}}_{2k+1}$, is as follows:
\begin{equation}
\begin{aligned}
\mu_{2k+1}[p]
=&\frac{2k(2k-1)B_\textup{IQ}^2}{P^4}\sum_{\rho\in\mathbb{P}}\Lambda^\textup{DL}[p+\rho]\mu_{2k-1}[\rho]\\
&+\frac{(k+1)^2B_\textup{IQ}^2}{P^4}\left|\mathbb{P}^\textup{DL}\right|^2\mu_{2k-1}[p],
\end{aligned}
\label{eq.thm.mu}
\end{equation}
where $B_\textup{IQ}=(1+(b^\textup{IQ})^2)A_\textup{digi}^2$.
\end{thm}
\begin{proof}
See Appendix B. 
\end{proof}
For the 3rd nonlinear basis, a perfect closed-form of $\mu_3$ exists, as the following proposition. 

\begin{prop}
The expectation of the 3rd nonlinear basis, $\mu_{3}$, in closed-form is as follows:
\begin{equation}
\mu_3[p]=\frac{B_\textup{IQ}^3}{P^4}\left(2\left|\mathbb{Q}^3_p\right|+4\left|\mathbb{P}^\textup{DL}\right|^2\right).
\end{equation}
\end{prop}
\begin{proof}
For (\ref{eq.thm.mu}), we can substitute $2k+1=1\rightarrow k=1$, and replace $\mu_1$ as follows:
\begin{equation}
\mu_1[p]=E\left[\left|\mathcal{X}_m^\text{IQ}[p]\right|^2\right]=
\begin{cases}
B_\text{IQ}=(1+|b^\text{IQ}|^2)A_\text{digi}^2\quad&p\in\mathbb{P}^\text{DL}\\
0\quad &p\notin\mathbb{P}^\text{DL}.
\end{cases}
\end{equation}
Then we can sum up with the equation as follows:
\begin{equation}
\begin{aligned}
\mu_3[p]&=\frac{2B_\textup{IQ}^2}{P^4}\sum_{\rho\in\mathbb{P}}\Lambda^\text{DL}[p+\rho]\mu_1[\rho]+\frac{4A_\text{digi}^4}{P^4}\left|\mathbb{P}^\text{DL}\right|^2\mu_{1}[p]\\
&=\frac{2B_\textup{IQ}^2}{P^4}\sum_{\rho\in\mathbb{P}^\text{DL}}\Lambda^\text{DL}[p+\rho]B_\textup{IQ}+\frac{4A_\text{digi}^4}{P^4}\left|\mathbb{P}^\text{DL}\right|^2B_\textup{IQ}.
\end{aligned}
\end{equation}
As $\sum_{\rho\in\mathbb{P}^\text{DL}}\Lambda^\text{DL}[p+\rho]=\left|\mathbb{Q}^3_p\right|$, the proposition follows.
\end{proof}
In the same step as Proposition~2.1, a closed-form for $\mu_{2k+1}$ after the $5$th order can also be derived. Fig.~\ref{fig.TheoVal} presents the validation of the theorems and the propositions compared with the numerical-based values.
\begin{figure}[t]
	\centering
	{\includegraphics[width=1.0\columnwidth]{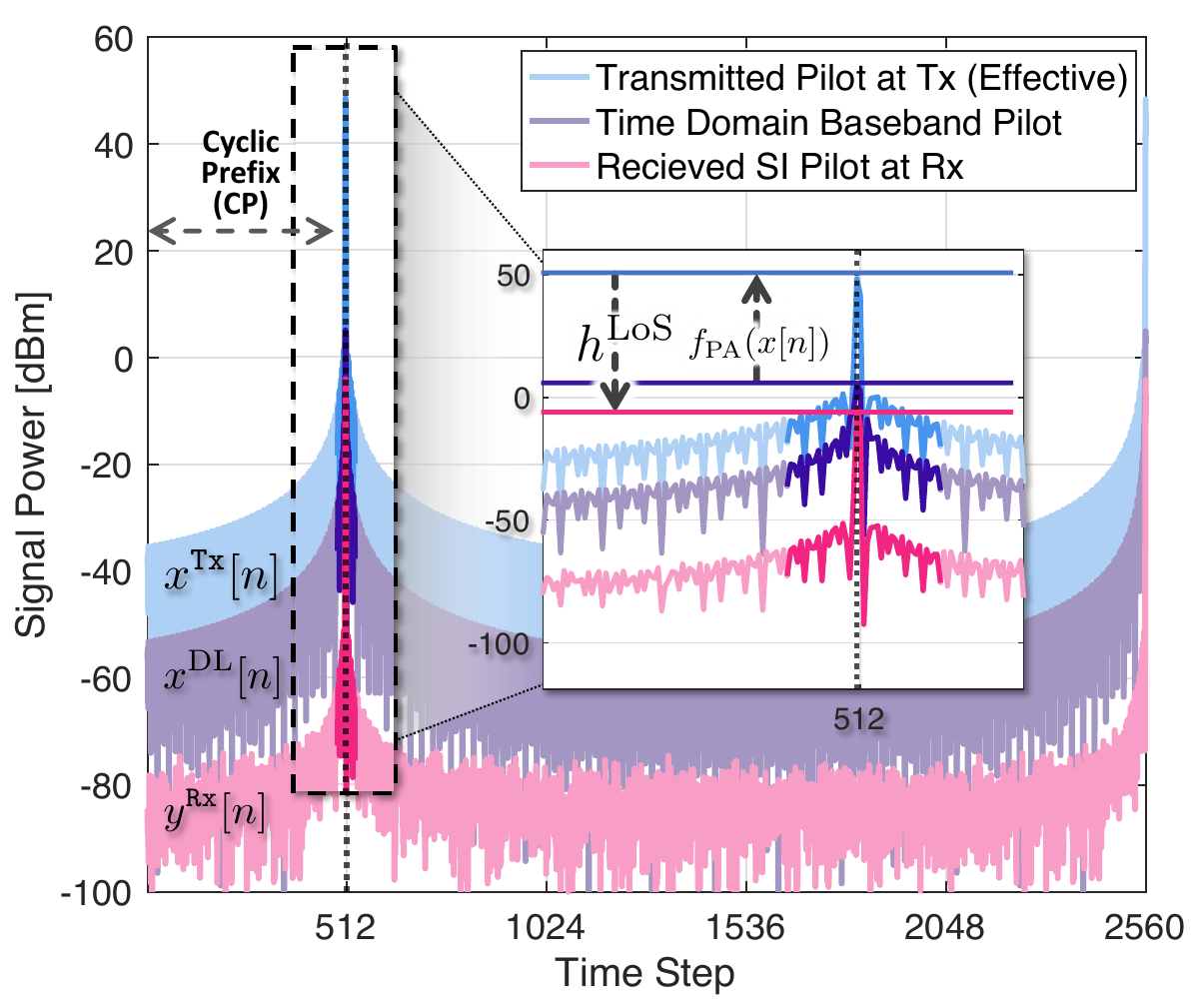}%
}
	\caption{The figure shows the procedure of $\mathtt{T/Rx}$ nonlinearity estimation, especially the power amplifier (PA) nonlinearity. The system estimate the nonlinearity with the proposed pilot pattern, $\mathcal{X}^\mathtt{imp}_m$.}
	\label{fig.algo_PA}
\vspace{-15pt}
\end{figure}

\subsection{Algorithm and Pilot Pattern for Low-complexity SIC: Reducing $\sum|\mathbb{K}_p|$}
By defining an impulse-like pilot signal for symbols $m\in\mathbb{M}^\mathtt{imp}$, our proposed method estimates the nonlinearity of the transmitter ($\mathtt{Tx}$) using beamforming information expressed as $h^\text{LoS}=\mathbf{w}^\text{UL}_\text{BS}\mathbf{H}^\text{LoS}\mathbf{f}^\text{DL}_\text{BS}$, depicted in Fig.~\ref{fig.algo_PA}. This process ensures that $\mathbb{M}^\mathtt{imp}$ is a subset of $\mathbb{M}^\mathtt{train}$. For symbols with $m\in\mathbb{M}^\mathtt{train}-\mathbb{M}^\mathtt{imp}$, the self-interference (SI) channel $\mathcal{H}^\text{SI}[p]$ is estimated. The system employs a novel pilot pattern to estimate nonlinearity, as outlined in the following proposition.
\begin{prop}
The $(2k+1)$th nonlinear basis of SI pilot, when the pilot allocated as $\mathcal{X}^\mathtt{imp}_m[p]=A_\textup{digi}e^{j\Omega p}$, is as follows:
\begin{equation}
\begin{aligned}
\mathit\Phi^\mathtt{imp}_{2k+1,m}[p]=\left|\mathbb{Q}^{2k+1}_p\right|A_\textup{digi}^{2k+1}|1+b^\textup{IQ}|^{2k}(1+b^\textup{IQ})e^{j\Omega p}.
\end{aligned}
\end{equation} 
\label{thm.SQ}
\end{prop}
\begin{proof}
As $\mathcal{X}_m^\mathtt{imp}[p]=e^{j\Omega p}$, all the components in the summation of $\mathit\Phi_{2k+1,m}$ in (\ref{eq.PhiUL}) has always same value as follows: 
\begin{equation}
\begin{aligned}
&\prod_{i=1}^{k+1}\mathcal{X}_m^\text{IQ}[q_i]\prod_{j=k+2}^{2k+1}\mathcal{X}^\text{IQ}_m[q_j]^*\\
&=\prod_{i=1}^{k+1}\left(\mathcal{X}_m[q_i]+b^\text{IQ}\mathcal{X}_m[-q_i]^*\right)\\
&\quad\quad\quad\times\prod_{j=k+2}^{2k+1}\left(\mathcal{X}_m[q_j]+b^\text{IQ}\mathcal{X}_m[-q_j]^*\right)^*\\
&=\prod_{i=1}^{k+1}A_\text{digi}\left(1+b^\text{IQ}\right)e^{j\Omega q_i}\prod_{j=k+2}^{2k+1}A_\text{digi}\left((1+b^\text{IQ})e^{j\Omega q_j}\right)^*\\
&=~A_\text{digi}^{2k+1}|1+b^\text{IQ}|^{2k}(1+b^\text{IQ})e^{j\Omega p}.
\end{aligned}
\end{equation}
As all components have the same phase and magnitude, we just multiply the size of the set we obtain in Theorem~1.
\end{proof}

To measure the PA coefficient separately, the BS utilizes analog beamforming information with $\mathbf{f}$ at $\mathtt{Tx}$ and $\mathbf{w}$ at $\mathtt{Rx}$. Both $\mathtt{Tx}$ and $\mathtt{Rx}$ maintain a fixed configuration within the BS. Once the UL/DL UE is scheduled and beam alignment is complete, the system is considered fully informed. The proposed method introduces a unique pilot pattern that enhances the algorithm for nonlinearity estimation. Generating a pilot signal waveform that resembles an impulse for OFDM is necessary due to the influence of a previous sampling time index in the cyclic prefix (CP). The system transmits $x^\mathtt{imp}_m$ at BS $\mathtt{Tx}$ and receives $y^\mathtt{imp}_m = y^\mathtt{Rx}_m[N_\text{CP}]$ at $\mathtt{Rx}$ across $|\mathbb{M}^\mathtt{imp}|$ symbols.

For the sampling time index $n < N_\text{CP}$, the signals pass through the NLOS SI channel, $h^\text{NLoS}$, and are returned as $y^\mathtt{Rx}_m[N_\text{CP}]$. If the system sufficiently suppresses $|x_m[n]|$ for all $n < N_\text{CP}$, then $\left|h^\text{NLoS} * x_m^\mathtt{Tx}[1:(N_\text{CP}-1)]\right|^2 < E\left[|z|^2\right]$, ensuring that the NLOS SI channel does not affect the SI channel estimation. Proposition 2.2 simplifies the process by calculating $\mathit\Phi_{2k+1,m}$ with constant complexity. The procedure for estimating nonlinearity is detailed from lines 1 to 7 in Algorithm~1.

\begin{algorithm} [!t]
	\caption{The Proposed SIC - Estimation Step}\label{algo.EST}
	\begin{algorithmic}[1]
	\Statex \textit{$\mathtt{T/Rx}$ nonlinearity, SI channel estimation,} 
	\Statex \textit{over} $m\in\mathbb{M}^\mathtt{train}$\textit{symbols}
		\State {\textbf{Input:} $x^\mathtt{imp}_m=x^\text{DL,CP}_m[N_\text{CP}], y^\mathtt{imp}_m=y^\mathtt{Rx}_m[N_\text{CP}],
		\mathcal{H}^\text{LoS}[p],$}
		\Statex {$\quad\quad\quad \mathcal{X}_m^\mathtt{train}[p],\mathcal{Y}_m^\mathtt{train}[p], \mathbb{Q}^{2k+1}_p, \mu_{2k+1}[p], \Gamma$}
		\State \textbf{Output:} $\hat{\mathcal{H}}^\text{SI}[p], b^\text{IQ}, \hat{\mathbf{a}}=\{a_{2k+1}\}, \mathbb{K}_p$
		\State Estimate $b^\text{IQ}$ from $\mathcal{X}_m[p], \mathcal{X}_m^*[-p],$ and $\mathcal{Y}_m[p]$ for $p\in\mathbb{P}^\text{DL}$
		\ForAll{$p\in\mathbb{P}^\text{UL}$}
		\ForAll{$m\in\mathbb{M}^\mathtt{imp}$}
		\State Obtain $\pmb{\mathit{\Phi}}_{m}^\mathtt{imp}[p]=\left[\cdots,\mathit{\Phi}_{2k+1,m}^\mathtt{imp}[p],\cdots\right]$ 
		\State from Proposition 2.2
		\State Append $\pmb{\mathit{\Phi}}_{m}^\mathtt{imp}[p], \mathcal{Y}^\mathtt{imp}_m[p]$ to $\pmb{\mathit{\Phi}}^\text{train}_p,\pmb{\mathcal{Y}}^\text{train}_p$ 
		\EndFor
		\State Obtain $\hat{\mathbf{a}}=\text{Estimator}\left(\mathcal{H}^\text{LoS}{\pmb{\mathit{\Phi}}}^\text{train}_p,\pmb{\mathcal{Y}}^\text{train}_p\right)$
		\EndFor
		\ForAll{$p\in\mathbb{P}^\text{UL}$}
		\ForAll{$m\in\mathbb{M}^\mathtt{train}$}
		\State $\mathit{\Phi}_{2k+1,m}^\mathtt{train}[p] = \text{IMD}\left(\mathcal{X}^\text{IQ}_{m},\mathit{\Phi}^\mathtt{train}_{2k-1,m}\right)$
		\State Append $\mathit{\Phi}_{2k+1,m}^\mathtt{train}[p]$ to $\pmb{\mathit\Phi}_p^\mathtt{train}$ 
		\EndFor
		\State $\hat{\mathcal{H}}^\text{SI}[p]=$Estimator$\left(\mathbf{\hat{a}}\pmb{\mathit\Phi}^\mathtt{train}_p,{\pmb{\mathcal{Y}}^\mathtt{train}_p}\right)$
		\While{$\hat{\mathcal{I}}_{2k+1}[p]>\Gamma,k<K_\text{max}$}
		\State $\hat{\mathcal{I}}_{2k+1}[p]=\hat{a}_{2k+1}^2\mu_{2k+1}[p]\left|\hat{\mathcal{H}}^\text{SI}[p]\right|^2$
		\State $k \rightarrow \mathbb{K}_p$	
		\EndWhile
		\EndFor
	\end{algorithmic}
\end{algorithm}

\begin{algorithm} []
	\caption{The Proposed SIC - Running Step}\label{algo.Run}
	\begin{algorithmic}[1]
	\Statex \textit{Running step for each sample,} $\forall m \in \mathbb{M}^\mathtt{run}, \forall p \in \mathbb{P}^\text{UL}$
		\State {\textbf{Input:} $\mathcal{Y}^\mathtt{Rx}[p],\mathcal{X}^\text{DL}[p],\hat{\mathcal{H}}^\text{SI}[p],\hat{\mathbf{a}}=\{\hat{a}_{2k+1}\}, b^\text{IQ}, \mathbb{K}_p$}
		\State \textbf{Output:} $\hat{\mathcal{Y}}^\text{UL}[p]$
		\State $\mathit{\Phi}_{1,m}[p]=\mathcal{X}^\text{IQ}_m[p]=\mathcal{X}_m[p]+b^\text{IQ}\left(\mathcal{X}_m[-p]\right)^*$
		\For{$k\in \mathbb{K}_p$}
		\State $\mathit{\Phi}_{2k+1,m}=\text{IMD}(\mathcal{X}^\text{IQ}_m,\mathit{\Phi}_{2k-1,m})$
		\EndFor
		\State $\hat{\mathcal{X}}^\mathtt{Tx}_m[p]=\left[\hat{a}_1\cdots\hat{a}_{2|\mathbb{K}_p|+1}\right]\left[\mathit{\Phi}_{1,m}[p]\cdots\mathit{\Phi}_{2|\mathbb{K}_p|+1,m}[p]\right]^\top$
				\State $\hat{\mathcal{Y}}^\text{UL}_m[p]=\mathcal{Y}^\mathtt{Rx}_m[p]-\hat{\mathcal{H}}^\text{SI}[p]\hat{\mathcal{X}}_m^\mathtt{Tx}[p]$

	\end{algorithmic}
\end{algorithm}

\subsection{Computational Complexity Analysis}

Table~\ref{table.complexity} compares the complexity per sample of our proposed SIC method with that of the conventional frequency domain SIC. As demonstrated in Algorithms 1 and 2, both the estimation and running steps require calculation of the nonlinear basis, $\mathit{\Phi}_{2k+1,m}[p]$. According to Proposition~2.2, for $m\in\mathbb{M}^\text{imp}$, the generation of $\mathit\Phi_{2k+1,m}$ can be computed with constant complexity. During the estimation of the SI coefficient, it is essential to minimize the size of the input matrix, such as $\pmb{\mathit{\Phi}}_p$. Conventional methods have complexity proportional to $P^\text{UL}(K\text{max}+1)$ because they estimate the PA coefficient and SI channel for each subcarrier simultaneously, represented as $\hat{\pmb{C}}_p$. For the method described in\cite{FDC_iter}, $R$ represents the number of iterations, with stable SIC performance expected when $R\geq3$. However, our proposed method achieves lower complexity by separately processing steps that require $\mathcal{O}(K_{\text{max}}+1)$ and $\mathcal{O}(P^\text{UL})$, when RF chain nonlinearity is frequency independent. Even when PA nonlinearity and IQ imbalance are frequency selective, they allow for a lower estimation repetititon compared to the SI channel. Consequently, only the SI channel requires frequent updates and estimation, which involves $\mathcal{O}(P^\text{UL})$.

The proposed method, utilizing two training types, $\mathbb{M}^\text{imp}$ and $\mathbb{M}^\text{train}$, achieves the lowest complexity. Even if an estimator other than LS is used, complexity increases proportionally with the size of the input matrix~\cite{K_FH}. Further reductions in complexity can be expected through adaptive selection of the nonlinearity order, $\left|\mathbb{K}_p\right|$, for each subcarrier.
	
\begin{figure*}
\begin{center}
    \subfigure[Distribution of SIC performance, IBFD]{%
      \includegraphics[width=0.66\columnwidth,keepaspectratio]
      {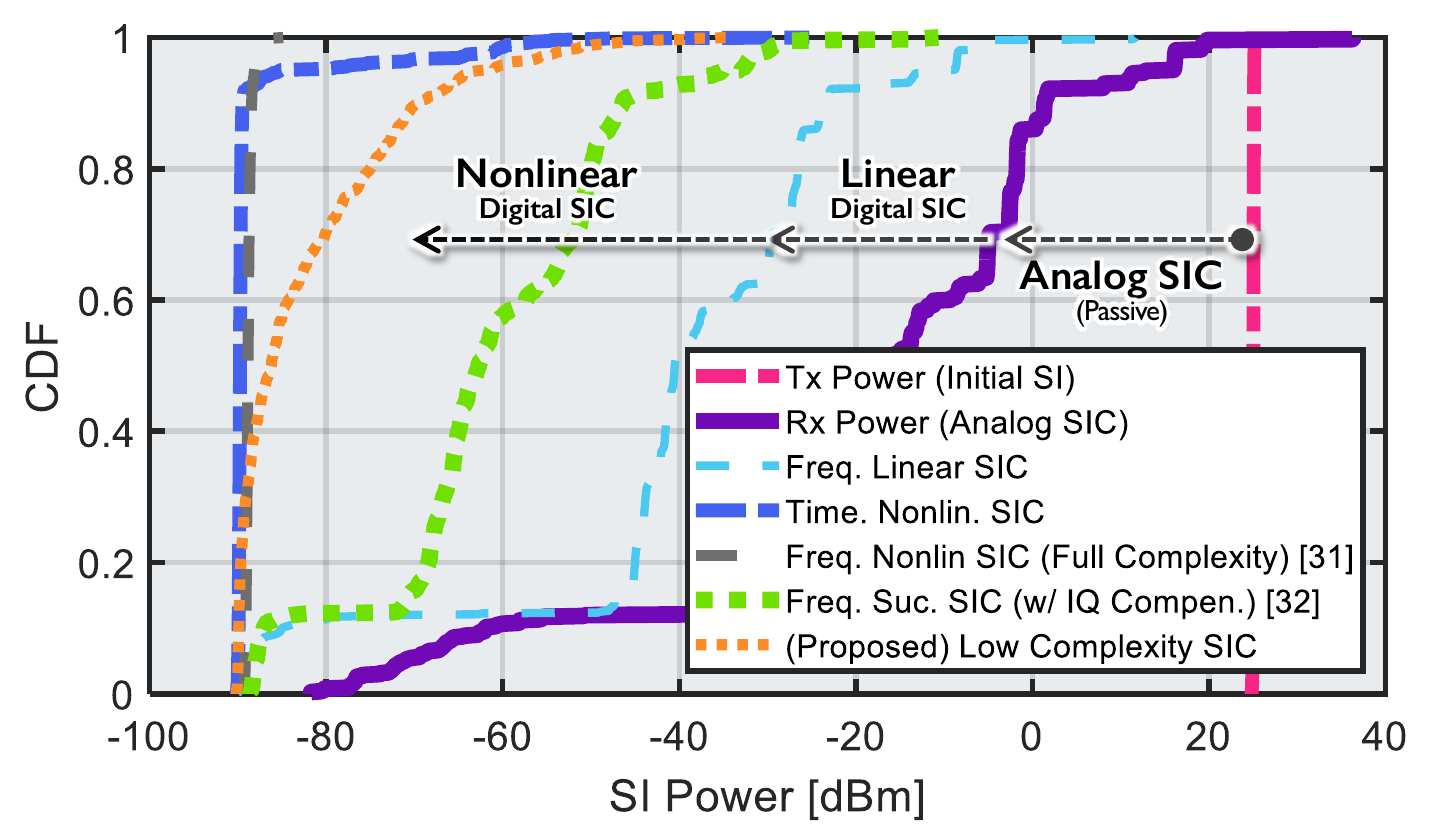}%
      \label{subfig.IB_cdf}
    }
    \subfigure[Distribution of SIC performance, SBFD]{%
      \includegraphics[width=0.66\columnwidth,keepaspectratio]
      {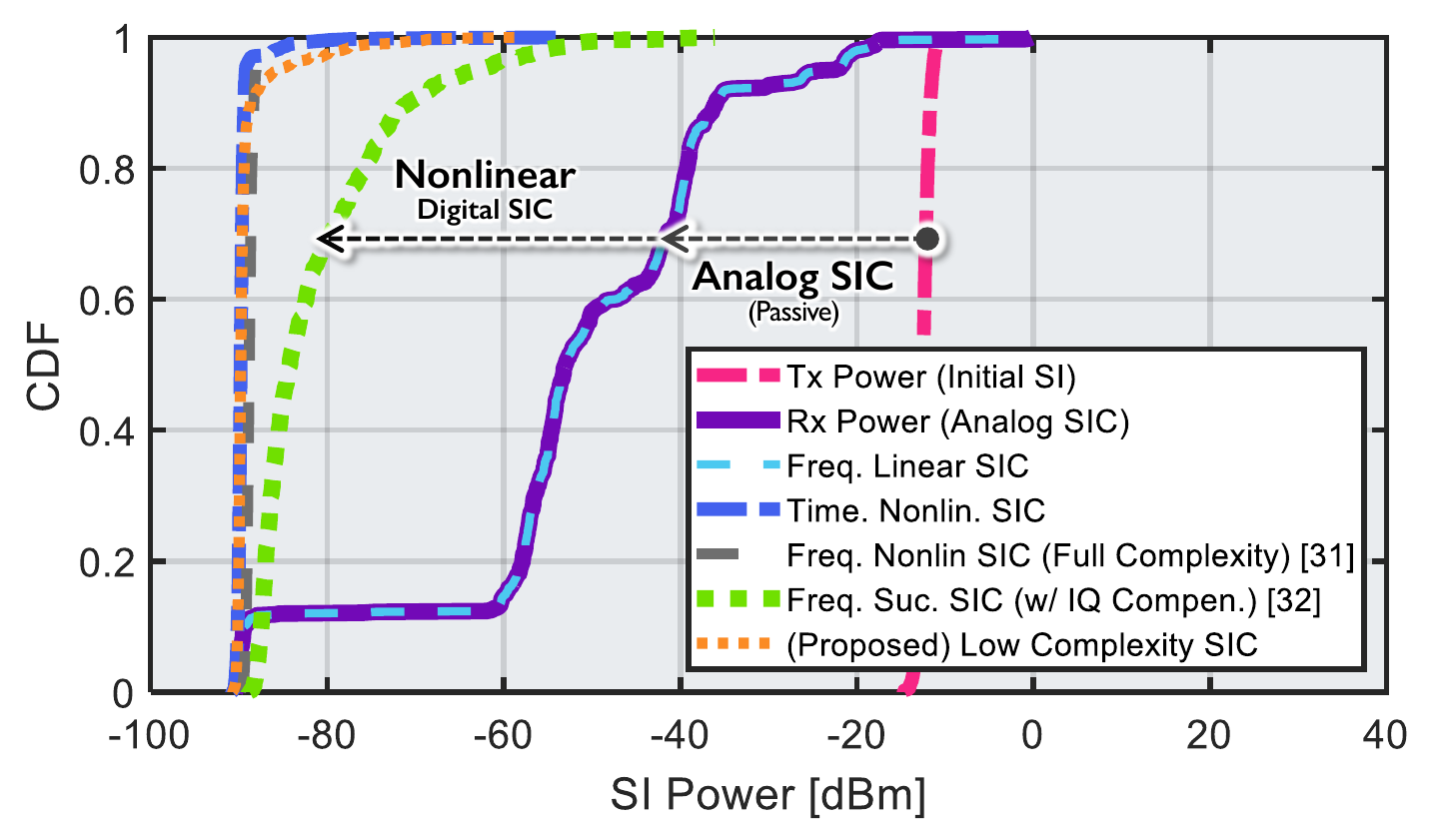}%
      \label{subfig.SB_cdf}
    }    
    \subfigure[Distribution of SIC performance, Partial overlap]{%
      \includegraphics[width=0.66\columnwidth,keepaspectratio]
      {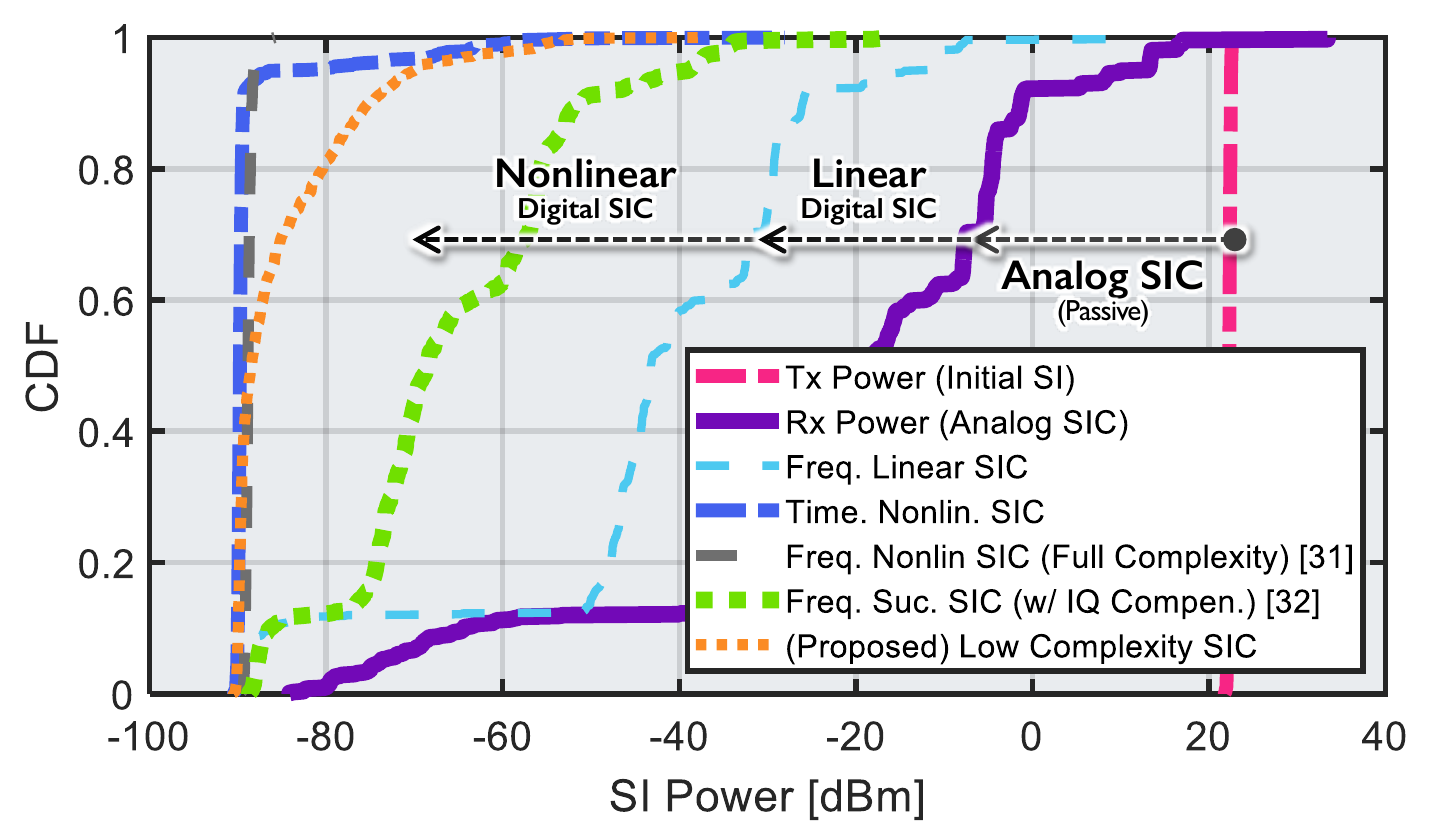}%
      \label{subfig.PO_cdf}
    }
        \subfigure[Residual SI power, IBFD]{%
      \includegraphics[width=0.66\columnwidth,keepaspectratio]
      {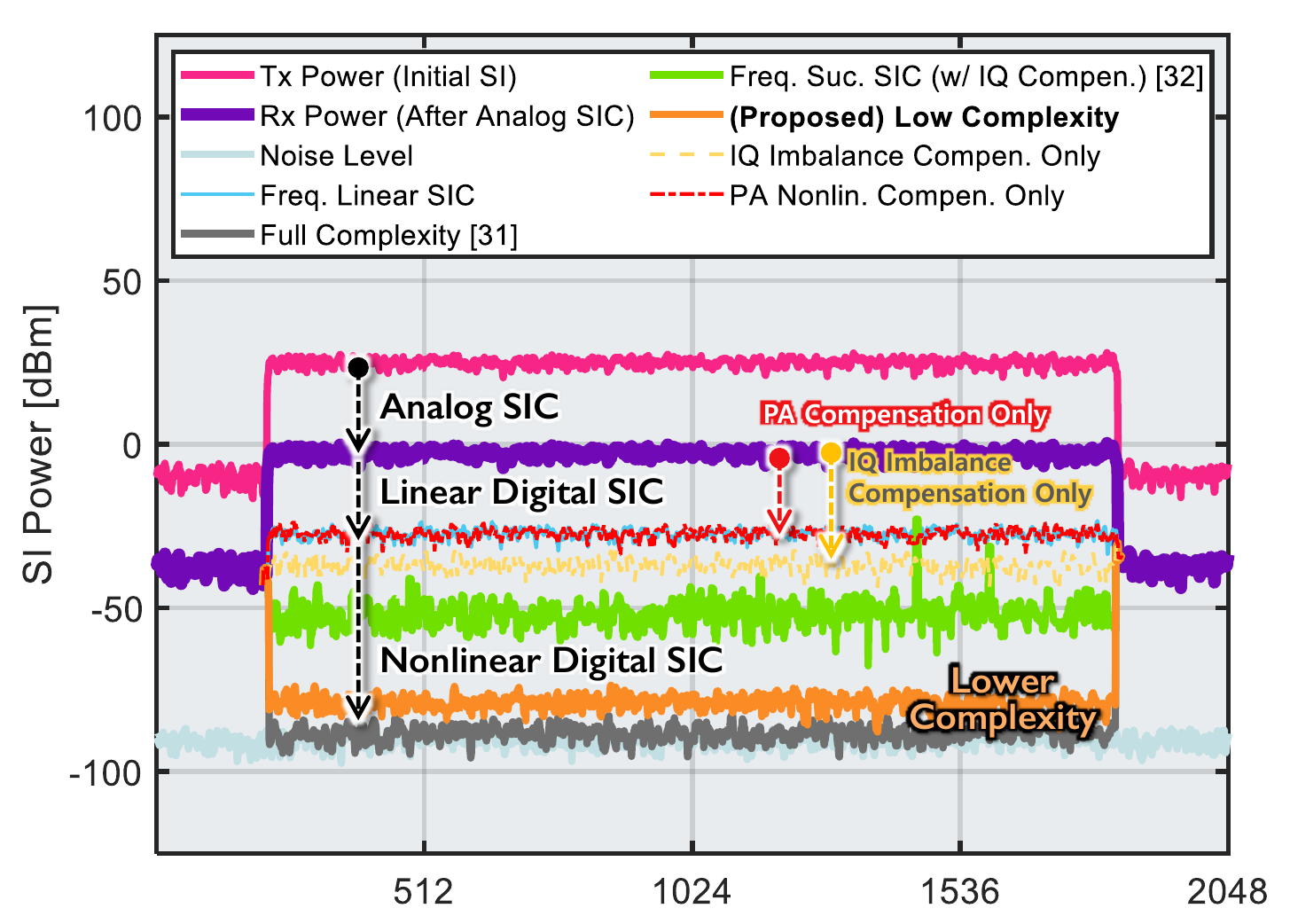}%
      \label{subfig.IB_RSI}
    }
    \subfigure[Residual SI power, SBFD]{%
      \includegraphics[width=0.66\columnwidth,keepaspectratio]
      {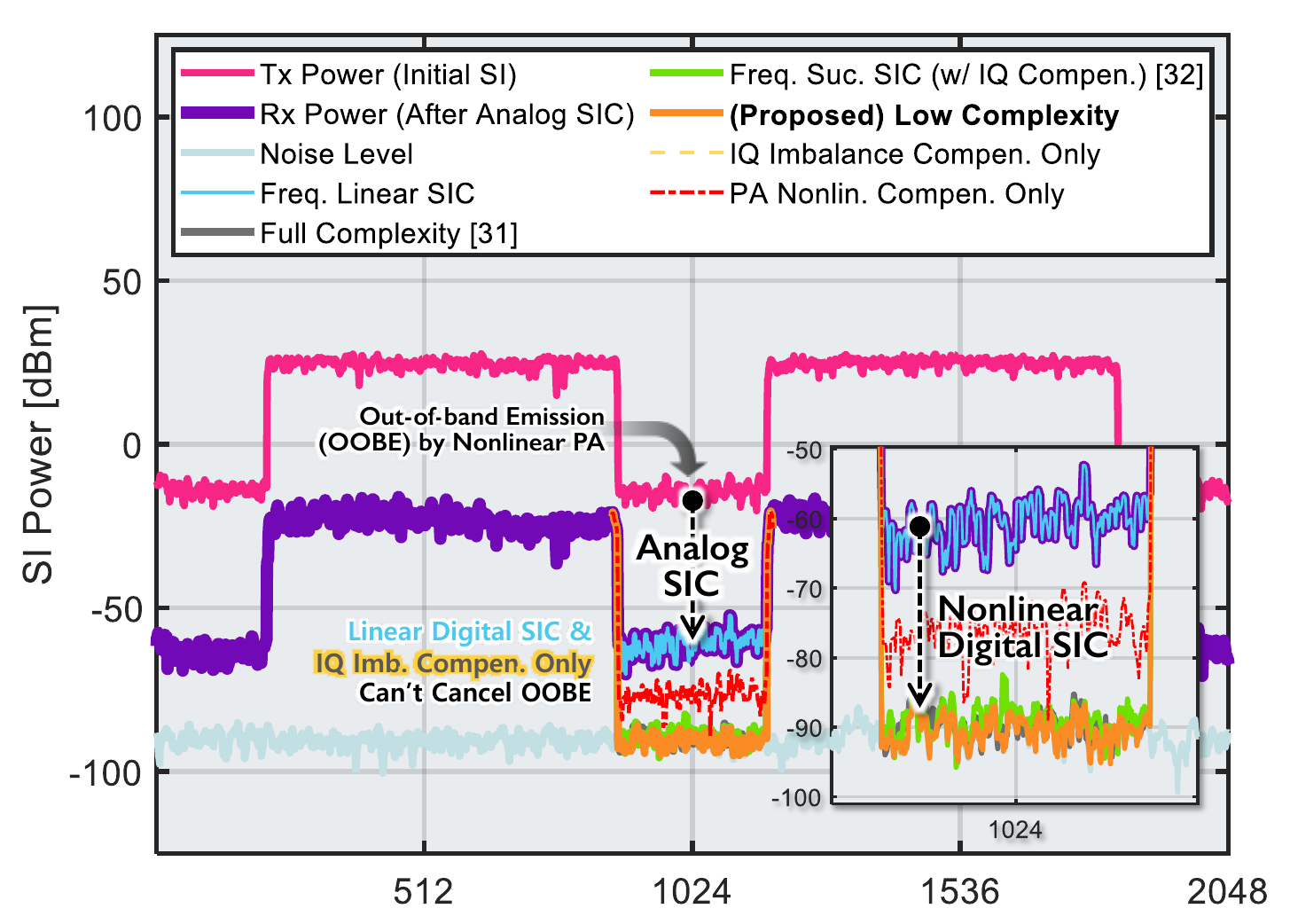}%
      \label{subfig.SB_RSI}
    }    
    \subfigure[Residual SI power, Partial overlap]{%
      \includegraphics[width=0.66\columnwidth,keepaspectratio]
      {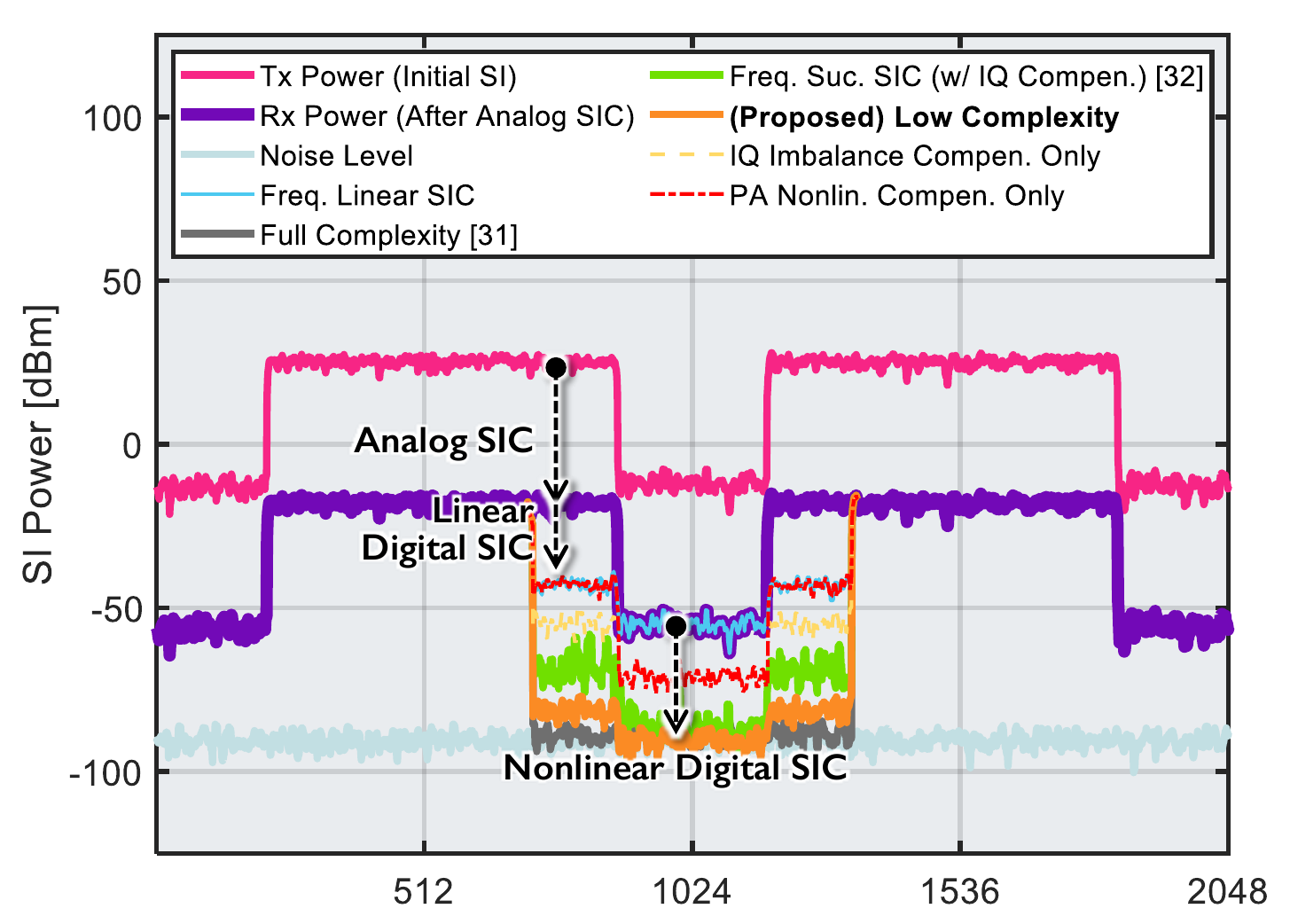}%
      \label{subfig.PO_RSI}
    }
 \end{center}
	\caption{Simulation results of nonlinear digital SIC for each duplex mode, in-band full-duplex (IBFD), sub-band full-duplex (SBFD), and UL/DL partial overlap. For each digital SIC algorithms, cumulative distribution function (CDF) in terms of the residual SI is presented in (a)--(c), and (d)--(f) depicts the residual SI in terms of the subcarrier index. }
	\label{fig.ResSIC}
\end{figure*}

\begin{center}
	\begin{table}[] 
		\caption{Simulation Parameters}
		\begin{tabular}{>{\centering } m{1.5cm} |>{\centering} m{3.5cm} |>{\centering} m{2.5cm} }
			\toprule
			\textbf{Paramet{\tiny }er} & \textbf{Description} & \textbf{Value}
			\tabularnewline
			\midrule
			\multicolumn{3}{c}{\textbf{Sub-6 GHz channel parameters  {(including SLS)}}}
			\tabularnewline\hline
			\centering		$f_\mathrm{c}$  & Carrier frequency & 3.5 (GHz)  \tabularnewline \hline
			\centering		$\Delta f$  & Subcarrier spacing & $60$ (kHz)  \tabularnewline \hline
			\centering		$E\left[|\mathcal{Z}[p]|^2\right]$  & Noise level & -90 (dBm)  \tabularnewline \hline
			\centering		$N_\text{Ray}$  & Number of rays through multipath & 25 \tabularnewline \hline
			\centering		-  &  {Maximum number of reflection for each ray} &  {10} \tabularnewline \hline
			\multicolumn{3}{c}{\textbf{BS specifications}} 
			\tabularnewline\hline
			\centering		$N^\mathtt{Tx}_\text{array}, N_\text{array}^\mathtt{Rx}$  & Antenna array  & $16$ $(4\times 4)$  UPA \tabularnewline \hline
			\centering		$P=|\mathbb{P}|$  & Number of subcarriers\\ (FFT size)& $2048$ \tabularnewline \hline
			\centering		$d_\text{ISO}$ & Isolation between $\mathtt{Tx}$ and $\mathtt{Rx}$ & $1~\text{(m)}$ \tabularnewline \hline

			\centering      $-$ & Modulation order & $ 16-\text{QAM}$\tabularnewline \hline
			\centering      $\left(a_1A_\text{digi}\right)^2$ & Tx power at BS\\(for each subcarrier)& $23$ (dBm)\tabularnewline \hline
			\centering      - &  {Tx power at UE\\ (Desired UL Signal)}&  {$15$ (dBm)}\tabularnewline \hline
			\centering      $\left|b^\text{IQ}\right|^2$ & IRR & $ 25$ (dB)\tabularnewline \hline
			\multicolumn{3}{c}{\textbf{SIC parmeters}} 
			\tabularnewline\hline
			\centering		$L$  & Number of time delay taps\\(Time domain SIC) & 20  \tabularnewline \hline
			\centering		$2K_\text{max}+1$  &  {SIC nonlinear order\\(Order of digital canceller)} & $ {5}$  \tabularnewline\hline
			\multicolumn{3}{c}{ {\textbf{Flexible Duplex Settings}}} \tabularnewline\hline
			\centering       {$\mathbb{P}^\text{UL}_\text{IB}=\mathbb{P}^\text{DL}_\text{IB}$} &  {UL/DL subcarrier set\\(IBFD)} &  {$\{1238,\cdots,2858\}$\\($135$ REs)} \tabularnewline\hline
			\centering       {$\mathbb{P}^\text{DL}_\text{SB}=\mathbb{P}^\text{DL}_\text{PO}$} &  {DL subcarrier set\\(SBFD \& Partial Overlap)} &  {$\{1238,\cdots,1904\}\cup\{2192,\cdots,2858\}$} \tabularnewline\hline
			\centering       {$\mathbb{P}^\text{UL}_\text{SB}$} &  {UL subcarrier set\\(SBFD)} &  {$\{1904,\cdots,2192\}$\\($24$ REs)} \tabularnewline\hline
			\centering       {$\mathbb{P}^\text{UL}_\text{PO}$} &  {UL subcarrier set\\(Partial overlap)} &  {$\{1742,\cdots,2354\}$\\($51$ REs)} \tabularnewline\bottomrule

		\end{tabular}
		\label{table.parameter}
	\end{table}
\end{center}

\section{Performance Results and Remarks}

\subsection{Simulation Purpose and Setups}
We conducted extensive simulations to evaluate the SIC performance of the proposed method of various flexible duplex scenarios. In the simulation, a single BS supports UL-UE and DL-UE via flexible duplex. The SI at BS, ${\mathcal{Y}}_m^\text{SI}[p]$, is occurred at the BS  $\mathtt{Rx}$, generated from $\mathcal{X}_m[p]$. In order to decode $\mathcal{Y}_m^\text{UL}[p]$, the BS employs SIC. We quantified the residual SI after operating the proposed digital SIC as $\sum_{p\in\mathbb{P}^\text{UL}}\left|{\mathcal{Y}}_m^\text{SI}[p]-\hat{\mathcal{Y}}_m^\text{SI}[p]\right|^2$. The carrier frequency was set to $3.5$~GHz. The number of OFDM subcarriers or DFT size followed the 5G NR standard~\cite{NR,NR2} with $P=2048$. Other parameters are detailed in Table~\ref{table.parameter}.

	\begin{figure}
	    \begin{center}
    		\includegraphics[width=0.9\columnwidth]{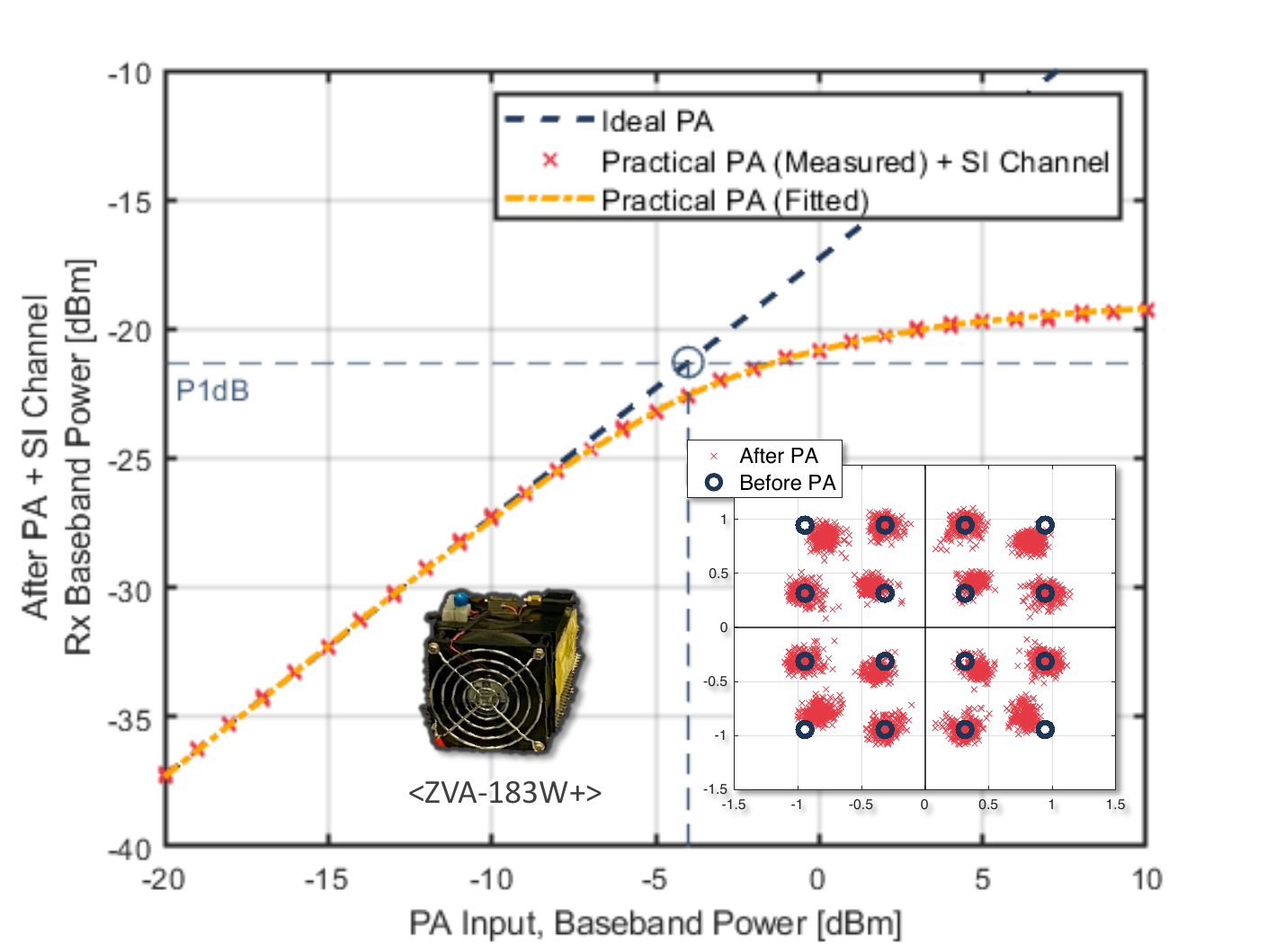}
	    \caption{ {PA nonlinearity and fitted data.}}
    	\label{fig.PAnon}
	\end{center}
	\end{figure}

 {To simulate the nonlinearity of PA, we measured the input and output of an actual PA, as illustrated in Fig.\ref{fig.testbed}, and represented it with the equation:
\begin{equation}
f_\text{PA}(x)=35.89x-2.24|x|^2x+0.0015|x|^4x.
\label{eq.PA}
\end{equation}
The nonlinear characteristics of the PA are presented in Fig.~\ref{fig.PAnon}, based on measurements from an actual PA. As the input power surpasses the P1dB point, a deviation from linear gain is observed in the Rx gain through the SI channel. The nonlinear model derived from the above equation shows a strong alignment with the measured data.
}
We generated the SI channel based on 3D ray-tracing. The time-domain SI channel, $\mathbf{H}^\text{SI}[n]$, consists of $N_\text{Ray}$ multipaths and is in length as $n\leq L$. We presumed the LoS channel, $\mathbf{H}^\text{SI}[n^\text{LoS}]=\mathbf{H}^\text{LoS}$, follows a near-field assumption~\cite{SphWave}, and that NLoS channels abide by far-field assumptions. Given these premises, we recorded the channel gain, $g_\ell$, AoA, $\theta^\text{AoA}_\ell$, and AoD, $\theta_\ell^\text{AoD}$ for observed rays, $\ell$. After collecting data for BSs in an outdoor environment, we generated a MIMO channel as in (\ref{eq.MIMOa}).

 {To evaluate the impact of the IQ imbalance $b^{\text{IQ}}$ on SIC performance and emphasize the importance of securing it in advance, we have incorporated these two items into the simulation results, as highlighted in the pilot approach discussed in Section IV.B. In the \textit{‘IQ imbalance compensation only'} case, only the IQ imbalance is addressed, with no nonlinear cancellation applied beyond the third term of the PA, resulting in only linear cancellation through compensation of  $b^{\text{IQ}}$. Conversely, in the \textit{‘PA nonlinearity compensation only’} case, compensation is applied solely to terms beyond the third term of the PA, without addressing the IQ imbalance. In this scenario, the nonlinear SI caused by $b^{\text{IQ}}\mathcal{X}^*_m[-p]$ in $\mathcal{X}^{\text{IQ}}_m[p] = \mathcal{X}_m[p] + b^{\text{IQ}} \mathcal{X}^*_m[-p]$ is left unmitigated, resulting in residual interference.
}

 {
We set the maximum value of the SIC nonlinearity order to 5. Although this value is unknown in advance at the SI canceller, we aligned it with the maximum order of the PA nonlinearity in (\ref{eq.PA}) to ensure a strict complexity comparison and avoid overestimation. We employed three duplex modes: IBFD, SBFD, and partially overlapped. Based on the 3GPP OFDM standard, we defined $\mathbb{P}^\text{UL}$ and $\mathbb{P}^\text{DL}$ according to the number of resource elements typically utilized with an FFT size of 2048.}

\begin{figure*}
\begin{center}
    \subfigure[ {Performance by the SIC complexity}]{%
      \includegraphics[width=0.9\columnwidth,keepaspectratio]
      {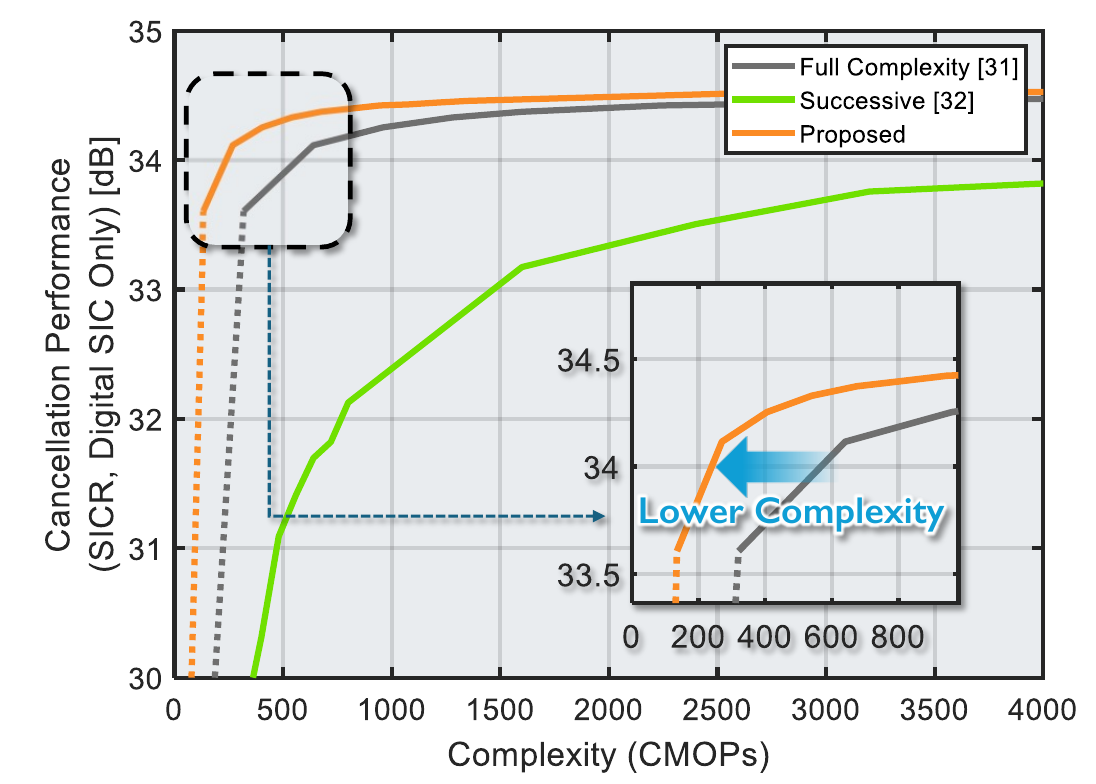}%
      \label{subfig.comp}
    }
    \subfigure[ {Performance by the number of training sample}]{%
      \includegraphics[width=0.9\columnwidth,keepaspectratio]
      {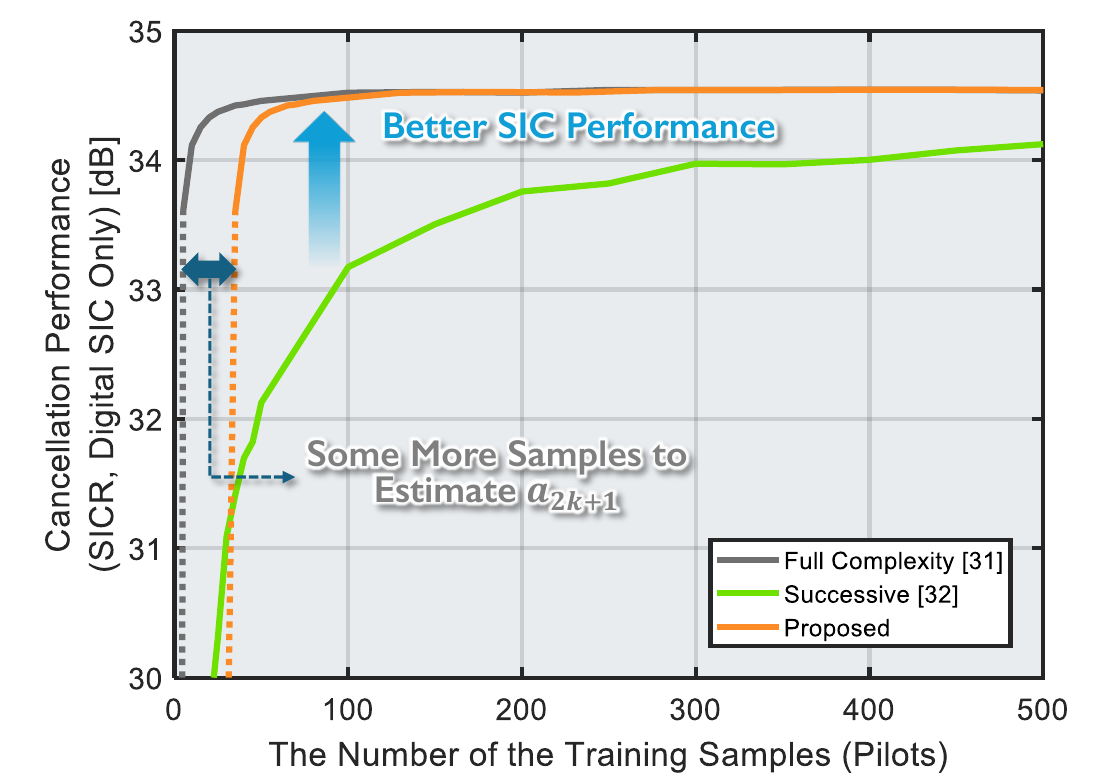}%
      \label{subfig.sample}
    }    
 \end{center}
	\caption{ {SIC performance tradeoff in terms of complexity and training samples.}}
	\label{fig.tradeoff}

\end{figure*}
 
\subsection{Performance of the Low-complexity SIC: Flexible Duplex Scenario}

Fig.\ref{fig.ResSIC} illustrates the SIC performance across different flexible duplex scenarios. The distribution of residual SI is displayed in Figs.\ref{subfig.IB_cdf}-\ref{subfig.PO_cdf}, and the residual SI for each subcarrier in the frequency domain is shown in Figs.\ref{subfig.IB_RSI}-\ref{subfig.PO_RSI}. After receiving SI following passive analog SIC, the BS operates the digital SIC in various modes. For frequency domain analysis, our proposed method is compared against linear SIC, a full LS-operating SIC~\cite{K_basis}, and a low-complexity SIC using a successive estimator~\cite{FDC_iter}.

Notably, our proposed SIC technique, despite being the least complex, approaches the ultimate performance limit: the noise level. While our method excludes some weak nonlinear basis from the estimation based on Theorem~2, which affects the marginal performance gap, this approach significantly reduces complexity. In scenarios with few DL subcarriers and no overlap with UL subcarriers, the SIC performance of our algorithm approaches the noise level. This is attributed to Theorem~2, where $\mu_{2k+1}[p]$ for $p\in\mathbb{P}^\text{UL}$ is proportional to $|\mathbb{P}^\text{DL}|$ and smaller if $p \notin \mathbb{P}^\text{DL}$, and $\hat{\mathcal{I}}_{2k+1}$ is smaller than the noise level, $E\left[|\mathcal{Z}[p]|^2\right]$.

In the IBFD case, as shown in Fig.\ref{subfig.IB_cdf}, the proposed method performs comparably to methods with higher complexity. Linear SIC, while effective, exhibits lower performance compared to nonlinear frequency domain technologies. For the SBFD scenario, where OOBE acts as SI due to nonlinearity, depicted in Fig.\ref{subfig.SB_RSI}, linear SIC fails to perform effectively unlike in the IBFD scenario.  {In the SBFD scenario, there is no linear SI component from DL to UL subcarriers, unlike in the IBFD or partial overlap cases. As a result, the initial SI in SBFD is the smallest and is solely due to nonlinear components. This is evident from the significantly lower initial SI in Fig. 5(b) compared to Figs. 5(a) and 5(c).}
The partial overlap case, depicted in Fig.\ref{subfig.PO_cdf}, shows intermediate performance improvements between IBFD and SBFD. In areas of UL/DL overlap, it performs SIC up to the OOBE level, as shown in Fig.\ref{subfig.PO_RSI}. Linear SIC does not mitigate the nonlinear SI, as observed in this scenario. The impact of this phenomenon and the performance of nonlinear SIC are discussed in PoC-based SLS.

 {Focusing on IQ imbalance, the effect of the proposed pilot pattern is as follows. In the IBFD scenario, estimating only the nonlinear coefficient of the PA without compensating for IQ imbalance performs worse than linear digital SIC that compensates only IQ imbalance. This is due to the gap in the nonlinear basis increasing as the difference between $\mathcal{X}^{\text{IQ}}_m$  and $\mathcal{X}_m$  rises to higher-order $2k + 1$ terms. This effect is illustrated in Fig. 5(d), where, if IQ imbalance is not compensated, the impact of $b^{\text{IQ}} \mathcal{X}^*_m[-p]$ remains significant, resulting in suboptimal performance for ‘PA nonlinearity compensation only’ at the linear SIC level. Conversely, in the SBFD scenario, where there is no DL linear component in the UL band, digital SIC is not feasible by compensating only IQ imbalance, as shown in Figs. 5(e) and 5(f). Here, the ‘IQ imbalance compensation only’ approach closely aligns with linear SIC, yielding minimal cancellation. However, for nonlinear SIC in the subband, compensation of IQ imbalance $b^{\text{IQ}}$ is essential to achieve cancellation down to the noise level, as demonstrated in the ‘PA nonlinearity compensation only’ case. Digital SIC can thus be effectively implemented with low complexity by adaptively considering RF chain impairment across various flexible duplex situations.}

 {\subsection{Performance Tradeoff in terms of the Complexity \& Overhead}}
 {The proposed method provides a clear complexity advantage, with only minimal overhead compared to conventional methods, requiring just a few extra training samples to estimate the PA coefficients separately. To demonstrate this, we included simulation results in Fig.~\ref{fig.tradeoff}. To focus primarily on the factors affecting digital SIC, the self-interference cancellation ratio (SICR) is defined as follows:
\begin{equation}
\text{SICR} =E\left( \frac{ \sum_{p\in\mathbb{P}^\text{UL}}\left|\mathcal{Y}_m^\text{SI}[p]\right|^2}{ \sum_{p\in\mathbb{P}^\text{UL}}\left|\Delta{\mathcal{Y}}_m^\text{SI}[p]\right|^2} \right)
\end{equation}
where $\Delta{\mathcal{Y}}_m^\text{SI}$ denotes the residual SI after each SIC step.}

 {The SICR performance in relation to the complexity of the proposed method is shown in Fig.~\ref{subfig.comp}, and the required overhead is illustrated in Fig.~\ref{subfig.sample}. In particular, this overhead arises from the need for training samples to estimate and utilize $a_{2k+1}$ in advance. However, this situation is infrequent, occurring primarily when the PA is turned on or off, and the actual overhead may be less than indicated in the simulation results in Fig.~\ref{subfig.sample}. The proposed method achieves the highest SIC performance with the least complexity because it has the lowest complexity per training sample as shown in Table~\ref{table.complexity}.
}


\begin{figure*}[t]
\centering
		\subfigure[FD PoC testbed]{\includegraphics[width=1.3\columnwidth,keepaspectratio]{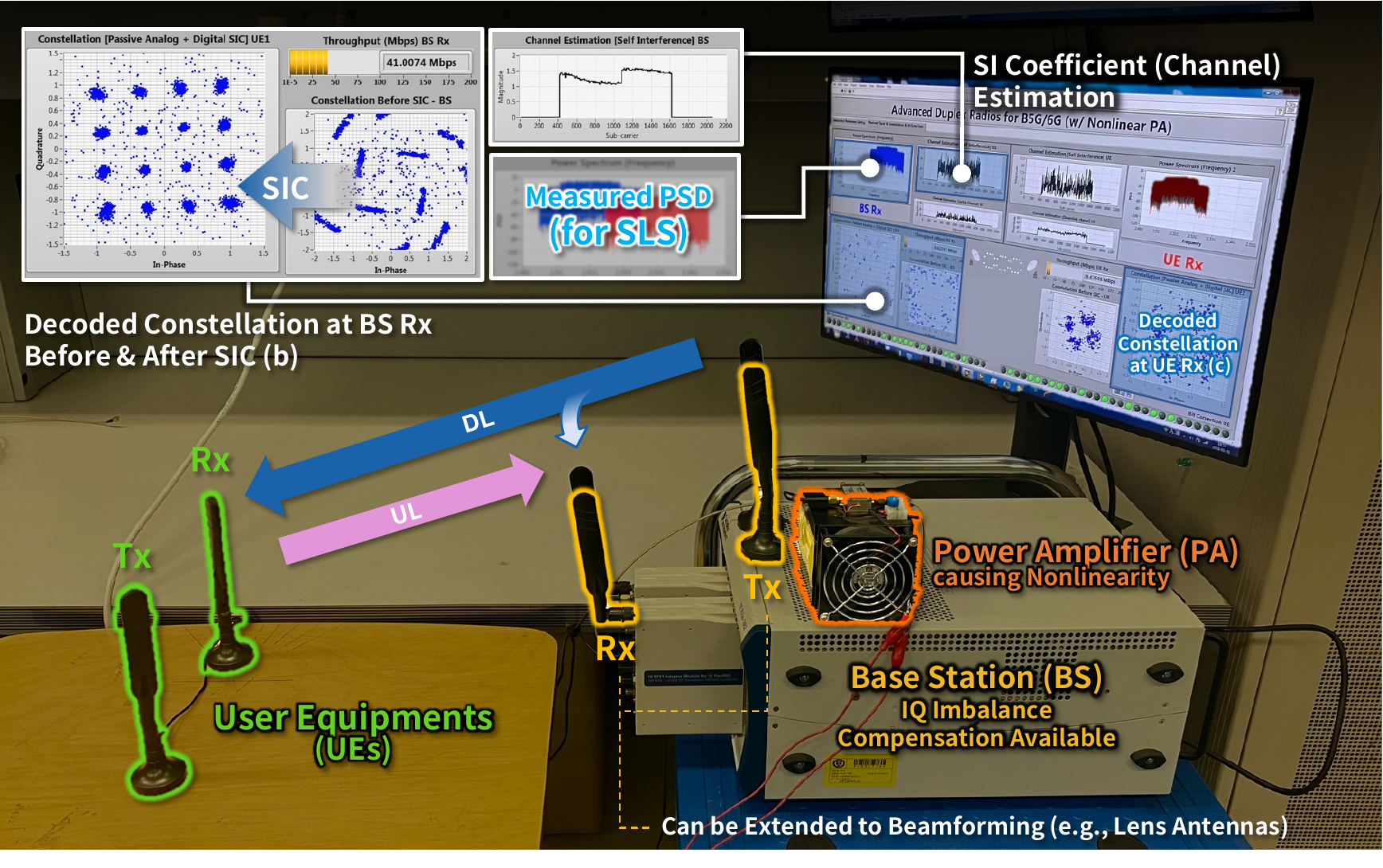}%
		\label{fig.testbed}}
		\subfigure[Received signal at each Rx antennas]{\includegraphics[width=.6\columnwidth,keepaspectratio]{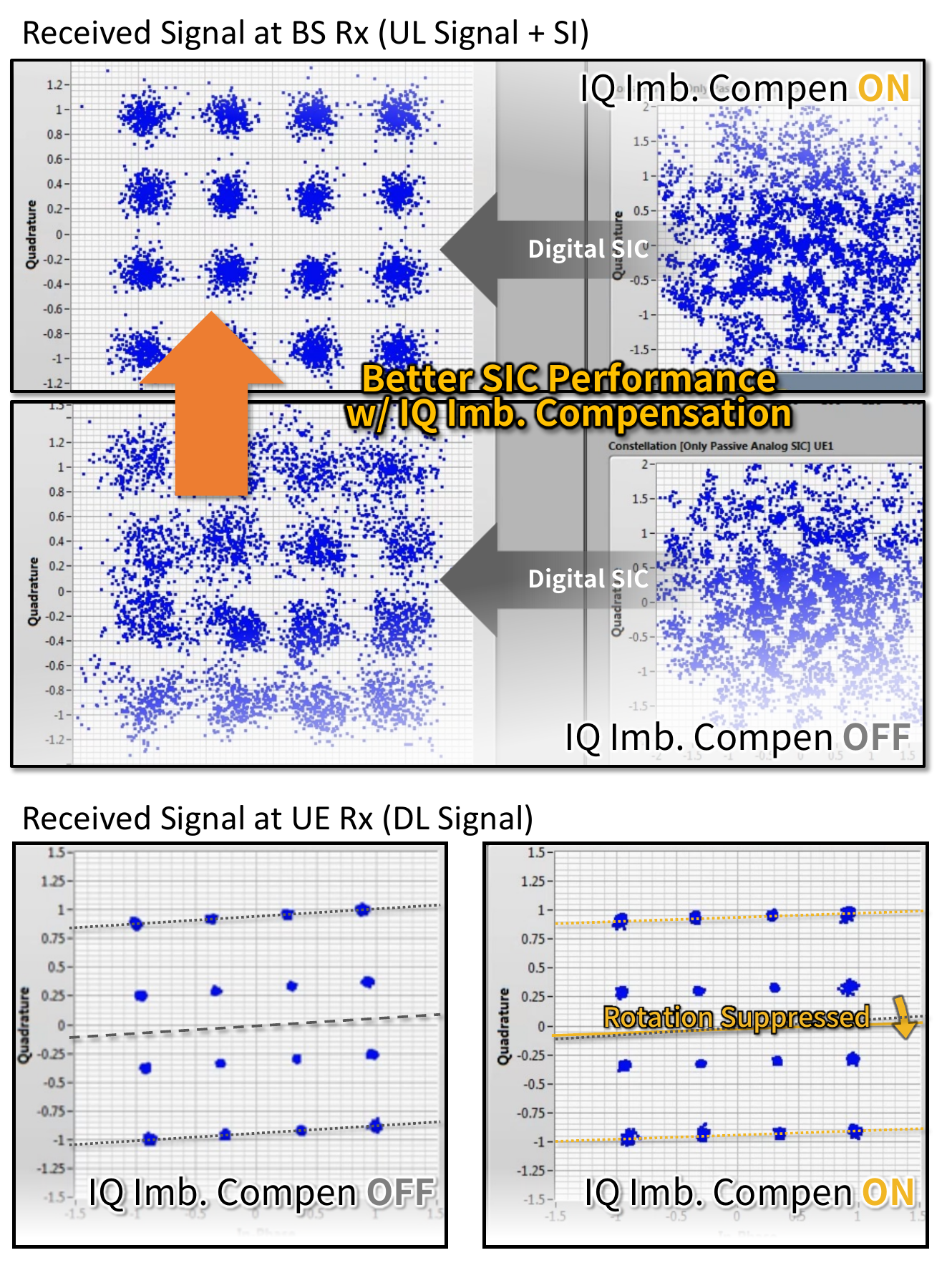}%
		\label{fig.IQsic}}
	\caption{ {PoC settings for flexible duplex and impact of PA nonlinearity and IQ imbalance.}}
\label{fig.poc}	
\end{figure*}

\begin{figure*}[t]
	\begin{center}
		\subfigure[ {The channels from BS to UE (DL) and UE to BS (UL) for system-level simulations (SLS) are measured using 3D ray-tracing.}]{\includegraphics[width=2.0\columnwidth,keepaspectratio]
			{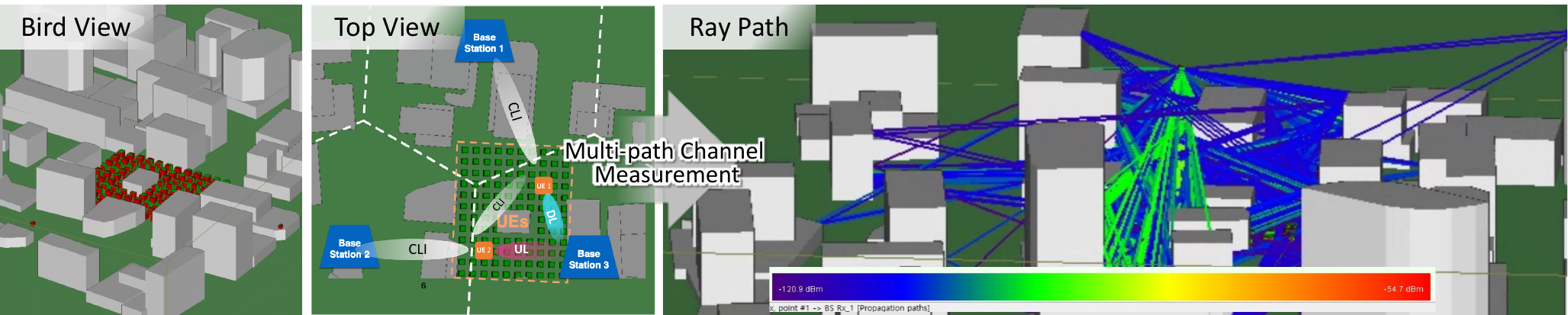}%
			\label{fig.raytracing}}
		\subfigure[ {FD SLS scenario}]{\includegraphics[width=.6\columnwidth,keepaspectratio]
			{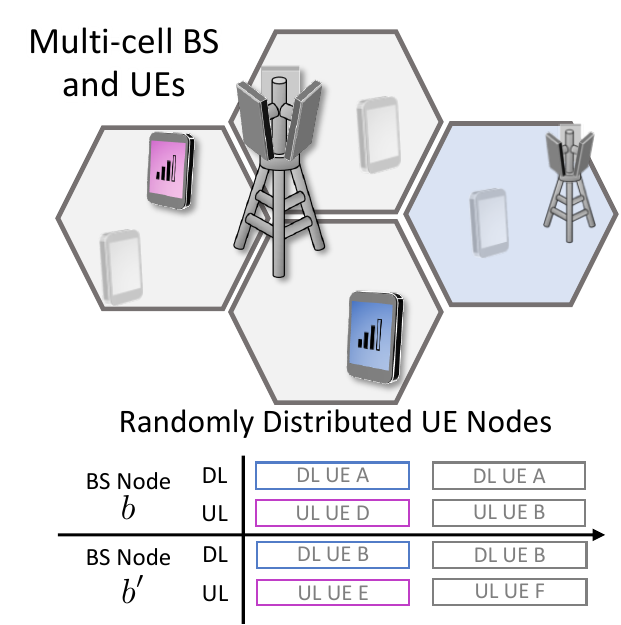}%
			\label{fig.te}}
		\subfigure[ {Throughput}]{\includegraphics[width=.7\columnwidth,keepaspectratio]
			{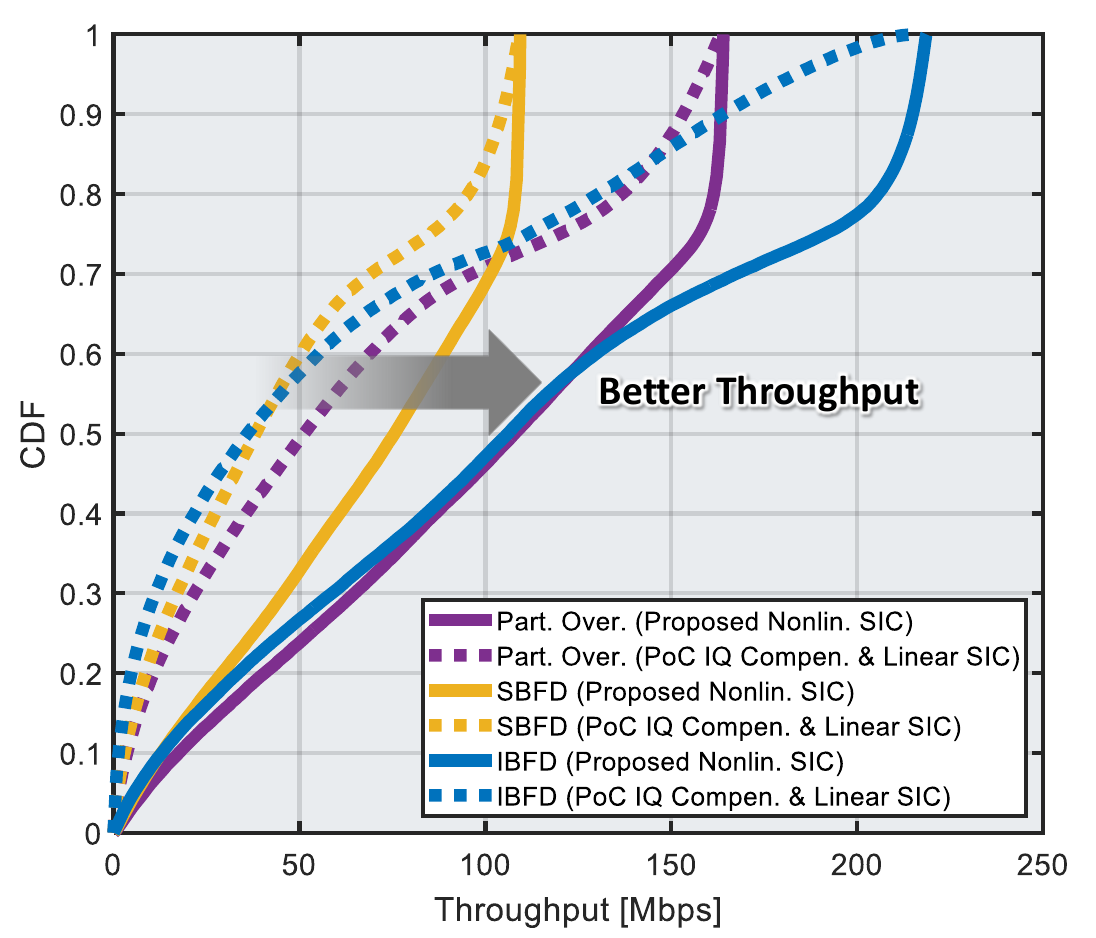}%
			\label{fig.sres1}}
		\subfigure[ {Cell coverage}]{\includegraphics[width=.7\columnwidth,keepaspectratio]
			{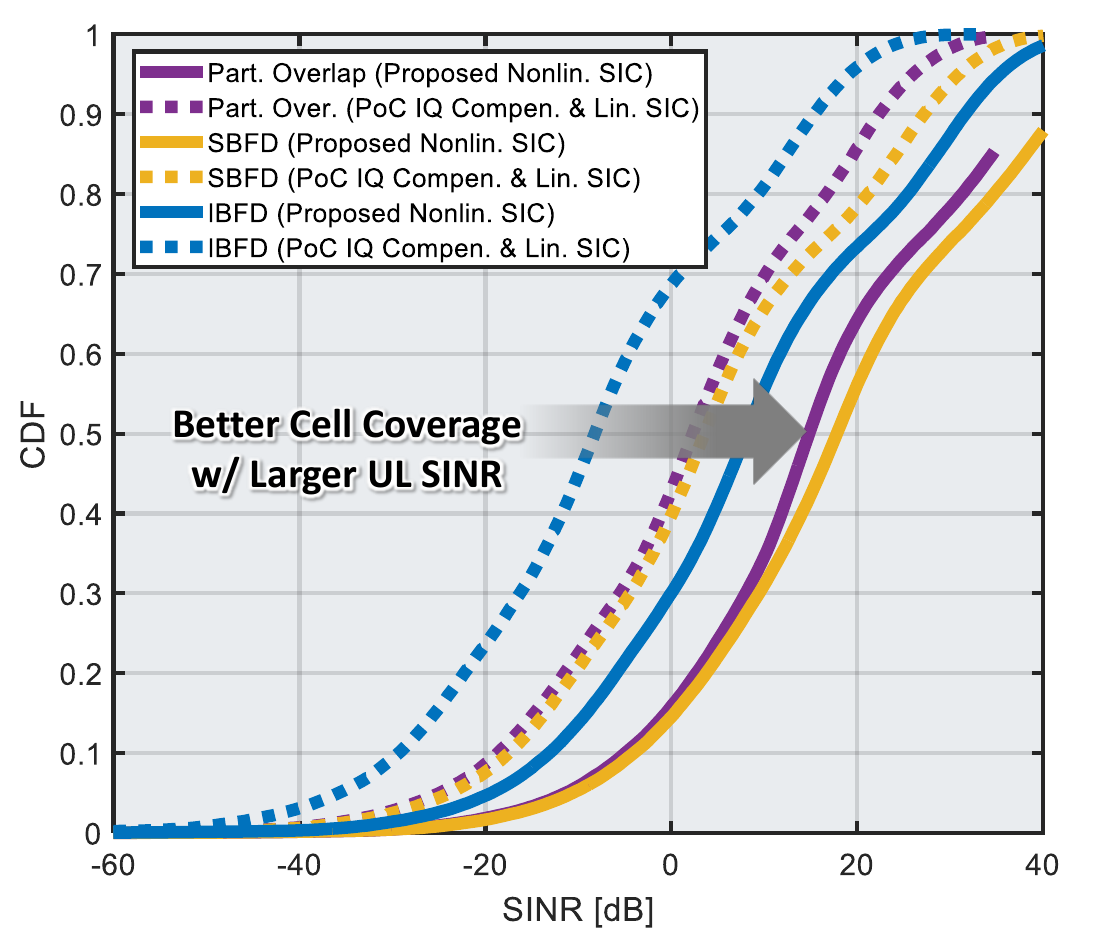}%
			\label{fig.sres2}}
		
	\end{center}
	\caption{The results for system-level simulation, considering the nonlinear SI measured by PoC, the multi-user scenario, 3D ray-tracing-based multipath channel, and the proposed digital SIC}
\label{fig.pores}	
\end{figure*}

\begin{figure*}[!b]
\hrulefill
\begin{subequations}
\begin{align}
\mathit\Phi_{2k+1,m}[p]
&=\frac{1}{P^{2k+1}}\left(\sum_{(q_1,\cdots,q_{2k+1})\in\mathbb{Q}^{2k+1}_p}\mathcal{X}_m[q_1]\mathcal{X}_m[q_2]\prod_{i=3}^{k+1}{\mathcal{X}_m[q_i]}\prod_{j=k+2}^{2k+1}{\mathcal{X}_m^*[q_j]}\right)\label{subeq.V.b}\\
&=\frac{1}{P^{2k+1}}\left\{\sum_{q_1,q_2\in\mathbb{P}^\text{DL},(q_3',\cdots,q_{2k+1}')\in \mathbb{Q}^{2k-1}_{p-(q_1+q_2)}}\mathcal{X}_m[q_1]\mathcal{X}_m[q_2]\left(\prod_{j=k+2}^{2k+1}{\mathcal{X}_m[q_j']}\prod_{i=3}^{k+1}{\mathcal{X}^*_m[q_i']}\right)^*\right\}\label{subeq.V.c}\\
&=\frac{1}{P^{2}}\left\{\sum_{q_1,q_2\in\mathbb{P}^\text{DL}}\mathcal{X}_m[q_1]\mathcal{X}_m[q_2]\left(\frac{1}{P^{2k-1}}\sum_{(q_3',\cdots,q_{2k+1}')\in \mathbb{Q}^{2k-1}_{p-(q_1+q_2)}}\prod_{j=k+2}^{2k+1}{\mathcal{X}_m[q_j']}\prod_{i=3}^{k+1}{\mathcal{X}^*_m[q_i']}\right)^*\right\}\label{subeq.V.d}\\
\mathit\Phi_{2k+1,m}[p]&=\frac{1}{P^{2k+1}}\left\{\sum_{(q_1,\cdots,q_{2k+1})\in\mathbb{Q}^{2k+1}_p}\mathcal{X}_m[q_1]\prod_{i=2}^{k+1}{\mathcal{X}_m[q_i]}\left(\mathcal{X}_m^*[q_{k+2}]\prod_{j=k+3}^{2k+1}{\mathcal{X}_m^*[q_j]}\right)\right\}\label{subeq.V2}
\end{align}
\end{subequations}
\end{figure*}

 {\subsection{PoC-based SLS Results: Impact of the Nonlinear SIC on Network Level} }

 {We evaluated the impact of the frequency domain, nonlinear SIC via SLS utilizing 3D ray-tracing and considering PoC implementation. We simulate the scheduling of DL-UE and UL-UE pairings in a multi-user scenario, confirming the throughput and the cell coverage performance across the network.} 

 {\subsubsection{PoC-based SIC measurement} Using PoC, we measured the frequency-domain SI levels under various flexible duplex PA nonlinearities, along with the residual SI after IQ imbalance correction and linear SIC, and subsequently reflected these measurements in the SLS and updated them relative to the initial submission. Additionally, we aimed to provide insights on scheduling DL-UE and UL-UE by combining UL/DL data measured through 3D ray-tracing and evaluating the proposed SIC performance.}
	
 {First, we investigated 1) whether SIC performance is affected by the presence or absence of IQ imbalance compensation, implemented as a default feature in our testbed, and 2) the extent to which IQ imbalance influences SIC performance directly. As shown in Fig.\ref{fig.IQsic}, compensating for IQ imbalance leads to an improvement in SIC performance. Additionally, we observed that the rotation of the DL signal, which can be interpreted as a signal similar to SI, is minimized at the UE Rx depending on the presence or absence of IQ imbalance compensation.}

\begin{center}
	\begin{table}[] 
		\caption{SLS Parameters}
		\centering
		\begin{tabular}{>{\centering} m{3.5cm} |>{\centering} m{2.5cm} }
			\toprule
			\textbf{Description} & \textbf{Value}
			\tabularnewline
			\midrule
			\centering		  {Inter-site distance (ISD)} &  {$150~\text{(m)}$} \tabularnewline \hline
			\centering		  {Number of UEs} &   {$100$} \tabularnewline \hline
			\centering		  {Size of UE grid} &   {$80$ (m) $\times$ $50$ (m)} \tabularnewline \bottomrule
		\end{tabular}
		\label{table.SLS}
	\end{table}
\end{center}

 {\subsubsection{3D ray-tracing-based channel generation}
Using 3D ray-tracing, we generated UL/DL multipath channels for randomly placed UEs. We measured UL/DL channels for three BSs and 100 UEs and included cross-link interference (CLI) to simulate realistic SLS, as shown in Fig.\ref{fig.raytracing}. The channel parameters are provided in Table \ref{table.parameter}. The SI, DL, and UL channels are all multipath and frequency-selective. For frequency-selective UL/DL link SINR, a representative value was determined using the effective SINR method~\cite{effSINR1}. Factors related to the spatial deployment of UEs and BSs are detailed in Table \ref{table.SLS}.}

 {Following the allocation of UL/DL, depicted in Figs.\ref{subfig.IB_RSI}-\ref{subfig.PO_RSI}, we assessed network performance after scheduling both linear SIC measured by PoC and proposed nonlinear SIC operations. As demonstrated in Fig.\ref{fig.testbed}, during the real-time operation of the flexible duplex PoC, we measured the SI received at BS $\mathtt{Rx}$ as power spectral density (PSD) in terms of frequency. These PSD measurements for IBFD, SBFD, and partial overlap were incorporated into the SLS and used to quantify the level of SI before digital SIC.}

 {\subsubsection{SLS results for FD networks}
Figs.\ref{fig.sres1} and \ref{fig.sres2} displayed improvements in throughput and UL SINR, respectively. With the proposed SIC, unlike linear SIC, the system can cancel OOBE, thus enhancing performance even in SBFD scenarios with orthogonal UL/DL allocation. While the use of only linear SIC in IBFD scenarios leaves significant residual SI, nonlinear SIC shows the best performance in these cases. Among flexible duplex scenarios, IBFD leads in throughput, and SBFD excels in cell coverage. This trade-off can be managed via partial overlap. }

 {The proposed SIC technology is low in complexity and, depending on the implementation approach, can achieve near-linear SIC complexity if $a_{2k+1}$ is secured in advance. Our SLS results illustrate the network-level performance advantages of the proposed method compared to linear SIC, particularly in terms of throughput and coverage. Specifically, the partial overlap duplex mode exhibits superior performance over IBFD and SBFD at higher frequencies. This result demonstrates that the proposed low-complexity, flexible frequency-domain nonlinear SIC approach has significant implications for enabling FD in base stations and networks, compatible with the current OFDM systems.}

\begin{figure*}[!b] 
\hrulefill
\begin{subequations}
\begin{align}
&E\left[\left|\mathit\Phi_{2k+1,m}[p]\right|^2\right]
=E\left[\mathit\Phi_{2k+1,m}[p]\left(\mathit\Phi_{2k+1,m}[p]\right)^*\right]\\
&=E\left[\frac{1}{P^{2k+1}}\left(\sum_{(q_1,\cdots,q_{2k+1})\in\mathbb{Q}^{2k+1}_p}\prod\mathcal{X}_m[q_i]\prod\mathcal{X}^*_m[q_j]\right)\frac{1}{P^{2k+1}}\left(\sum_{(q_1',\cdots,q_{2k+1}')\in\mathbb{Q}^{2k+1}_p}\prod\mathcal{X}_m[q_i']\prod\mathcal{X}^*_m[q_j']\right)^*\right]\label{subeq.E.b}\\
&=\frac{1}{P^{4k+2}}E\left[\sum_{\alpha,\beta,\gamma,\delta\in [1,2k+1]}\sum_{\tilde{q}_\alpha,\tilde{q}_\beta,\tilde{q}'_\gamma,\tilde{q}'_\delta\in\mathbb{P}^\text{DL}}
\left(\mathcal{X}_m[\tilde{q}_\alpha]\mathcal{X}_m[\tilde{q}_\beta]\mathcal{X}_m^*[\tilde{q}'_\gamma]\mathcal{X}_m^*[\tilde{q}'_\delta]\right)
\prod_{i\in[1,2k+1],i\neq \alpha,\beta}{\mathcal{X}_m[\tilde{q}_i]}\prod_{j\in[1,2k+1],j\neq \gamma,\delta}{\mathcal{X}_m^*[\tilde{q}'_j]}
\right]\label{subeq.E.c}\\
&=\frac{1}{P^4}\sum_{\gamma,\delta\in [1,2k], \gamma\neq \delta}\left(\sum_{\tilde{q}_\gamma'=q_1,\tilde{q}'_\delta=q_2}E\left[
\underbrace{\left(\mathcal{X}_m[q_1]\mathcal{X}_m[q_2]\mathcal{X}_m^*[\tilde{q}_\gamma']\mathcal{X}_m^*[\tilde{q}'_\delta]\right)}
_{E\left[\mathcal{X}_m[q_1]\mathcal{X}_m[q_2]\mathcal{X}_m^*[q_1]\mathcal{X}_m^*[q_2]\right]=A_\text{digi}^4}
\mathit\Phi_{2k-1,m}^*[\rho]\mathit\Phi_{2k-1,m}[\rho]\right]\right)\label{subeq.E.d}\\
&\quad\quad\quad+\frac{1}{P^4}\sum_{\alpha,\gamma\in[1,k+1]}\left(\sum_{\tilde{q}_\alpha=q_{2k+1},\tilde{q}_\gamma'=q_{2k+1}'}E\left[
\underbrace{\left(\mathcal{X}_m[\tilde{q}_\alpha]\mathcal{X}_m^*[q_{2k+1}]\right)}_{E[\mathcal{X}_m\mathcal{X}_m^*]=A_\text{digi}^2}\mathit\Phi_{2k-1,m}[\rho]
\underbrace{\left(\mathcal{X}_m^*[\tilde{q}_\gamma']\mathcal{X}_m[q_{2k+1}']\right)}_{E[\mathcal{X}_m^*\mathcal{X}_m]=A_\text{digi}^2}\mathit\Phi^*_{2k-1,m}[\rho]
\right]\right) \label{subeq.E.e}\\
&=\left(\frac{A_\text{digi}}{P}\right)^4\sum_{\gamma,\delta\in [1,2k], \gamma\neq \delta}\sum_{\tilde{q}_\gamma'=q_1,\tilde{q}'_\delta=q_2}E\left[\left|\mathit\Phi_{2k-1,m}[\rho]\right|^2\right]
+\left(\frac{A_\text{digi}}{P}\right)^4\sum_{\alpha,\gamma\in[1,k+1]}\sum_{\tilde{q}_\alpha=q_{2k+1},\tilde{q}_\gamma'=q_{2k+1}'}E\left[\left|\mathit\Phi_{2k-1,m}[p]\right|^2\right]
\label{subeq.E.f}\\
&=\left(\frac{A_\text{digi}}{P}\right)^4 2k(2k-1)\sum_{\rho=q_1+q_2-p}E\left[\left|\mathit\Phi_{2k-1,m}[\rho]\right|^2\right]
+\left(\frac{A_\text{digi}}{P}\right)^4(k+1)^2\left|\mathbb{P}^\text{DL}\right|^2E\left[\left|\mathit\Phi_{2k-1,m}[p]\right|^2\right]\label{subeq.E.g}
\end{align}
\label{eq.EQA}
\end{subequations} 
\end{figure*}

\section{Conclusion}
This paper proposed a nonlinear, low-complexity, frequency-domain, digital self-interference cancellation (SIC) algorithm for flexible duplex communications, covering both in-band and sub-band full-duplex scenarios. The effectiveness of this method, which decomposes features into nonlinear, linear, unchanged, and weak categories based on theorems and propositions, is demonstrated by our numerical results. Through multi-user, proof-of-concept-based simulations, we also showed the impact of implementing the proposed digital SIC with low-complexity. Furthermore, our method is proven to be compatible with OFDM systems.  {Future research will focus on network-level optimization of the flexible duplex system, including the development of scheduling algorithms, and implementing MIMO flexible duplex PoC.}

 \appendices
\section*{Acknowledgment}

The authors would like to thank Mr. Y.-T. Kim and Mr. S.-M. Kim for helpful discussions on PoC implementations.

\section{PROOF OF LEMMA 1}
The formulation of $\mathit\Phi_{2k+1,m}[p]$ is given by equation (\ref{subeq.V.b}). Moreover, equation (\ref{subeq.V.c}) presents the effective decomposition of $\mathbb{Q}^{2k+1}_p$ by $q_1, q_2\in\mathbb{P}^\text{DL}$, and $\mathbb{Q}^{2k-1}_\rho$. Subsequently, $\mathit\Phi_{2k+1,m}$ is decomposed by $\mathcal{X}_m[q_1], \mathcal{X}_m[q_2]$ and $\mathit\Phi^*_{2k-1,m}[\rho]$ as in (\ref{subeq.V.d}). It is given that $q_1+q_2-\rho=p$ as the definition of $\mathbb{Q}^{2k+1}_p$ as follows:
\begin{equation}
q_1+q_2-\left(-\sum_{i=3}^{{k+1}}q_i+\sum_{j={k+2}}^{{2k+1}}q_j\right)=q_1+q_2-\rho=p.
\end{equation}
Therefore, we can express $\rho$ as $\rho=q_1+q_2-p$ leading to the recursive form in equation (\ref{eq.lem1.a}).

For the alternative term, $\mathit\Phi_{2k+1,m}[p]$ is decomposed by $\mathcal{X}_m[q_1], \mathcal{X}_m[q_{k+2}]$ and $\mathit\Phi_{2k-1,m}[\rho]$ as in (\ref{subeq.V2}). The sum of $q_i$ fulfills $q_1-q_{k+2}+\rho=p$ for $\mathbb{Q}^{2k-1}_\rho$. By replacing $\rho$ with $p-q_1+q_{k+2}$, we arrive at equation (\ref{eq.lem1.b}), and the lemma then follows.

\section{PROOF OF THEOREM 2}
By rearranging the formula of (\ref{subeq.E.b}), we can get the equation (\ref{subeq.E.c}), where $\tilde{q}_\ell$ and $\tilde{q}'_\ell$ are as follows:
\begin{equation}
\tilde{q}_\ell=
\begin{cases}
q_\ell&\ell\in[1,k+1]\\
q_\ell'&\ell\in[k+2,2k]
\end{cases}\quad
\tilde{q}_\ell'=
\begin{cases}
q_\ell'&\ell\in[1,k+1]\\
q_\ell&\ell\in[k+2,2k]
\end{cases}.
\end{equation}
The last subcarrier index $q_{2k+1}$ and $q_{2k+1}'$, however, are automatically decided by $q_1,\cdots,q_{2k}$ and $q_1',\cdots,q_{2k}'$ to satisfy (\ref{eq.setQ}).
To calculate $\mu_{2k+1}$, we can utilize the correlation of the nonlinear mapping as follows: 
\begin{equation}
E\left[\mathit\Phi_{2k-1,m}[\rho]\mathit\Phi_{2k-1,m}^*[\rho']\right]=
\begin{cases}
\mu_{2k-1}&\forall \rho=\rho'\\
0&\forall \rho\neq\rho'
\end{cases}.
\label{eq.Erho}
\end{equation}
The multiplication value of $\mu_{2k-1}$ only needs to be known for the case where $\rho=\rho'$.
According to the results of Lemma~1, in the computation of $\mu_{2k+1}$, the effective expectation equivalent to a non-zero value can be represented as in equations (\ref{subeq.E.d}) and (\ref{subeq.E.e}). Consequently, we can consider two cases for $\mu_{2k-1}$.

In the first case, let $\tilde{q}_\gamma'=q_1,\tilde{q}'_\delta=q_2$ as in (\ref{subeq.E.d}). Given $\gamma,\delta\in [1,2k]$ ($q_{2k+1}$ fixed), and $\gamma\neq \delta$, the number of possible cases to select $\gamma,\delta$ is $2k(k-1)$. To filter the cases where 
\begin{equation}
E\left[\mathcal{X}_m[\tilde{q}_\alpha]\mathcal{X}_m[\tilde{q}_\beta]\mathcal{X}^*_m[\tilde{q}_\gamma]\mathcal{X}_m^*[\tilde{q}_\delta]\right]=~A_\text{digi}^2,
\end{equation}
the condition $q_1+q_2-\rho=p$ should be met. We can then calculate the summation with $\Lambda^\text{DL}[p+\rho]$ as in Theorem~1.

In the second case, with $q_{2k+1}$, and $q_{2k+1}'$ fixed, we have $\tilde{q}_\alpha=~q_{2k+1},\tilde{q}_\gamma'=q_{2k+1}'$ as in (\ref{subeq.E.e}). We select $\alpha,\gamma~\in[1,k+1]$, then the number of the cases is $(k+1)^2$. Moreover, each of $\tilde{q}_\alpha,\tilde{q}'_\gamma$ varies as $\tilde{q}_\alpha,\tilde{q}'_\gamma\in\mathbb{P}^\text{DL}$ and the number of pairs is $\left|\mathbb{P}^\text{DL}\right|^2$. 
Upon organizing the summation terms, we arrive at equation (\ref{subeq.E.f}), and the theorem follows.

\appendices

\bibliographystyle{IEEEtran}
\bibliography{Ref_TWC}

\begin{thebibliography}{10}
\providecommand{\url}[1]{#1}
\csname url@samestyle\endcsname
\providecommand{\newblock}{\relax}
\providecommand{\bibinfo}[2]{#2}
\providecommand{\BIBentrySTDinterwordspacing}{\spaceskip=0pt\relax}
\providecommand{\BIBentryALTinterwordstretchfactor}{4}
\providecommand{\BIBentryALTinterwordspacing}{\spaceskip=\fontdimen2\font plus
\BIBentryALTinterwordstretchfactor\fontdimen3\font minus
  \fontdimen4\font\relax}
\providecommand{\BIBforeignlanguage}[2]{{%
\expandafter\ifx\csname l@#1\endcsname\relax
\typeout{** WARNING: IEEEtran.bst: No hyphenation pattern has been}%
\typeout{** loaded for the language `#1'. Using the pattern for}%
\typeout{** the default language instead.}%
\else
\language=\csname l@#1\endcsname
\fi
#2}}
\providecommand{\BIBdecl}{\relax}
\BIBdecl

\bibitem{FD6G}
B.~Smida \emph{et~al.}, ``Full-duplex wireless for {6G}: {Progress} brings new
  opportunities and challenges,'' \emph{{IEEE} J. Sel. Areas Commun.}, vol.~41,
  no.~9, pp. 2729--2750, Sep. 2023.

\bibitem{advD}
M.~Giordani \emph{et~al.}, ``Toward {6G} networks: {Use} cases and
  technologies,'' \emph{{IEEE} Commun. Mag.}, vol.~58, no.~3, pp. 55--61, Mar.
  2020.

\bibitem{GuestChae}
G.~Fodor \emph{et~al.}, ``Guest editorial: {Full} duplex communications theory,
  standardization, and practice,'' \emph{{IEEE} Wireless Commun.}, vol.~28,
  no.~1, pp. 10--11, Feb. 2021.

\bibitem{3gppDupEn}
\emph{Study on evolution of {NR} duplex operation}.\hskip 1em plus 0.5em minus
  0.4em\relax document 3GPP TR 38.858 ver 18.1, Apr. 2024.

\bibitem{kim2015survey}
D.~Kim, H.~Lee, and D.~Hong, ``A survey of in-band full-duplex transmission:
  {From} the perspective of {PHY} and {MAC} layers,'' \emph{{IEEE} Commun.
  Surveys Tuts.}, vol.~17, no.~4, pp. 2017--2046, 2015.

\bibitem{FDnet_hwi}
Y.~Kim \emph{et~al.}, ``A state-of-the-art survey on full-duplex network
  design,'' \emph{Proc. of IEEE \textit{(Early Access)}}, pp. 1--24, 2024.

\bibitem{alpha1}
A.~AlAmmouri \emph{et~al.}, ``In-band $\alpha $ -duplex scheme for cellular
  networks: {A} stochastic geometry approach,'' \emph{{IEEE} Trans. Wireless
  Commun.}, vol.~15, no.~10, pp. 6797--6812, Oct. 2016.

\bibitem{US-alpha}
I.~Randrianantenaina, H.~Dahrouj, H.~Elsawy, and M.-S. Alouini, ``Interference
  management in {Full}-duplex cellular networks with partial spectrum
  overlap,'' \emph{IEEE Access}, vol.~5, pp. 7567--7583, 2017.

\bibitem{flex1}
K.~I. Pedersen \emph{et~al.}, ``A flexible {5G} frame structure design for
  frequency-division duplex cases,'' \emph{{IEEE} Commun. Mag.}, vol.~54,
  no.~3, pp. 53--59, Mar. 2016.

\bibitem{CBF}
C.-B. {Chae} \emph{et~al.}, ``Coordinated beamforming for the multiuser {MIMO}
  broadcast channel with limited feedforward,'' \emph{{IEEE} Trans. Signal
  Process.}, vol.~56, no.~12, pp. 6044--6056, Dec. 2008.

\bibitem{RISNOMA_hwi}
Y.~Kim, J.-H. Kim, and C.-B. Chae, ``Partition-based {RIS}-assisted multiple
  access: {NOMA} decoding order perspective,'' \emph{{IEEE} Trans. Veh.
  Technol.}, vol.~71, no.~8, pp. 9083--9088, Aug. 2022.

\bibitem{FD}
A.~Sabharwal \emph{et~al.}, ``In-band {Full}-duplex wireless: Challenges and
  opportunities,'' \emph{{IEEE} J. Sel. Areas Commun.}, vol.~32, no.~9, pp.
  1637--1652, Sep. 2014.

\bibitem{FDchae}
M.~Chung \emph{et~al.}, ``Prototyping real-time full duplex radios,''
  \emph{{IEEE} Commun. Mag.}, vol.~53, no.~9, pp. 56--63, Sep. 2015.

\bibitem{XDD}
H.~Ji \emph{et~al.}, ``Extending {5G TDD} coverage with {XDD}: Cross division
  duplex,'' \emph{IEEE Access}, vol.~9, pp. 51\,380--51\,392, 2021.

\bibitem{analSIC}
E.~Everett, A.~Sahai, and A.~Sabharwal, ``Passive self-interference suppression
  for {Full}-duplex infrastructure nodes,'' \emph{{IEEE} Trans. Wireless
  Commun.}, vol.~13, no.~2, pp. 680--694, Feb. 2014.

\bibitem{ASICjw}
J.~W. Kwak \emph{et~al.}, ``Analog self-interference cancellation with
  practical {RF} components for full-duplex radios,'' \emph{{IEEE} Trans.
  Wireless Commun.}, vol.~22, no.~7, pp. 4552--4564, Jul. 2023.

\bibitem{ASICjsac}
C.~W. Morgenstern \emph{et~al.}, ``Analog-domain self-interference cancellation
  for practical multi-tap full-duplex system: {Theory}, modeling, and
  algorithm,'' \emph{{IEEE} J. Sel. Areas Commun.}, vol.~41, no.~9, pp.
  2796--2807, Sep. 2023.

\bibitem{IanBSI}
I.~P. Roberts, A.~Chopra \emph{et~al.}, ``Beamformed self-interference
  measurements at 28 {GHz}: Spatial insights and angular spread,'' \emph{{IEEE}
  Trans. Wireless Commun.}, vol.~21, no.~11, pp. 9744--9760, 2022.

\bibitem{IanMag}
I.~P. Roberts, J.~G. Andrews \emph{et~al.}, ``Millimeter-wave full duplex
  radios: New challenges and techniques,'' \emph{{IEEE} Wireless Commun.},
  vol.~28, no.~1, pp. 36--43, Feb. 2021.

\bibitem{NLSI}
E.~Ahmed, A.~M. Eltawil, and A.~Sabharwal, ``Self-interference cancellation
  with nonlinear distortion suppression for {Full}-duplex systems,'' in
  \emph{Proc. Asilomar Conf. on Signal, Syst. and Comput.}, Nov. 2013, pp.
  1199--1203.

\bibitem{K_iter}
K.~Komatsu, Y.~Miyaji, and H.~Uehara, ``Iterative nonlinear self-interference
  cancellation for in-band {Full}-duplex wireless communications under mixer
  imbalance and amplifier nonlinearity,'' \emph{{IEEE} Trans. Wireless
  Commun.}, vol.~19, no.~7, pp. 4424--4438, Jul. 2020.

\bibitem{IanFreq}
I.~P. Roberts, H.~B. Jain, and S.~Vishwanath, ``Frequency-selective beamforming
  cancellation design for millimeter-wave full-duplex,'' in \emph{Proc. IEEE
  Int. Conf. on Commun. (ICC)}, Jun. 2020, pp. 1--6.

\bibitem{NLSIsim}
M.~S. Sim \emph{et~al.}, ``Nonlinear self-interference cancellation for
  {Full}-duplex radios: From link-level and system-level performance
  perspectives,'' \emph{{IEEE} Commun. Mag.}, vol.~55, no.~9, pp. 158--167,
  Sep. 2017.

\bibitem{ASICfreq}
Z.~H. Hong \emph{et~al.}, ``Frequency-domain {RF} self-interference
  cancellation for in-band full-duplex communications,'' \emph{{IEEE} Trans.
  Wireless Commun.}, vol.~22, no.~4, pp. 2352--2363, Apr. 2023.

\bibitem{K_FH}
K.~Komatsu, Y.~Miyaji, and H.~Uehara, ``Frequency-domain {Hammerstein}
  self-interference canceller for in-band {Full}-duplex {OFDM} systems,'' in
  \emph{Proc. IEEE Wireless Commun. and Netw. Conf. (WCNC)}, Mar. 2017, pp.
  1--6.

\bibitem{FDRs}
D.~Bharadia, E.~McMilin, and S.~Katti, ``Full duplex radios,'' in
  \emph{Proceedings of the ACM SIGCOMM}, 2013, pp. 375--386.

\bibitem{WideLinFD}
D.~Korpi \emph{et~al.}, ``Widely linear digital self-interference cancellation
  in direct-conversion full-duplex transceiver,'' \emph{{IEEE} J. Sel. Areas
  Commun.}, vol.~32, no.~9, pp. 1674--1687, Sep. 2014.

\bibitem{SIC_PN}
X.~Quan \emph{et~al.}, ``Impacts of phase noise on digital self-interference
  cancellation in full-duplex communications,'' \emph{{IEEE} Trans. Signal
  Process.}, vol.~65, no.~7, pp. 1881--1893, Apr. 2017.

\bibitem{SIC_PN2}
R.~Li, A.~Masmoudi, and T.~Le-Ngoc, ``Self-interference cancellation with
  nonlinearity and phase-noise suppression in full-duplex systems,''
  \emph{{IEEE} Trans. Veh. Technol.}, vol.~67, no.~3, pp. 2118--2129, Mar.
  2018.

\bibitem{compactFDd2d}
M.~Chung \emph{et~al.}, ``Compact full duplex {MIMO} radios in {D2D} underlaid
  cellular networks: From system design to prototype results,'' \emph{IEEE
  Access}, vol.~5, pp. 16\,601--16\,617, 2017.

\bibitem{K_basis}
K.~Komatsu, Y.~Miyaji, and H.~Uehara, ``Basis function selection of
  frequency-domain {Hammerstein} self-interference canceller for in-band
  {Full}-duplex wireless communications,'' \emph{{IEEE} Trans. Wireless
  Commun.}, vol.~17, no.~6, pp. 3768--3780, Jun. 2018.

\bibitem{FDC_iter}
Y.~He \emph{et~al.}, ``Frequency-domain successive cancellation of nonlinear
  self-interference with reduced complexity for {Full}-duplex radios,''
  \emph{{IEEE} Trans. Commun.}, vol.~70, no.~4, pp. 2678--2690, Apr. 2022.

\bibitem{202212access}
X.~Han \emph{et~al.}, ``Interference mitigation for non-overlapping sub-band
  full duplex for {5G}-advanced wireless networks,'' \emph{IEEE Access},
  vol.~10, pp. 134\,512--134\,524, Dec. 2022.

\bibitem{SBFD_hwi}
Y.~Kim, K.~Lee, J.~Jung, and C.-B. Chae, ``Frequency domain non-linear {SIC}
  for {B5G/6G} sub-band full-duplex,'' \emph{Proc. IEEE Glob. Commun. Conf.
  (GLOBECOM) Workshops}, pp. 1087--1092, Dec. 2023.

\bibitem{IQ_tsp}
M.~Valkama, M.~Renfors, and V.~Koivunen, ``Advanced methods for {I/Q} imbalance
  compensation in communication receivers,'' \emph{{IEEE} Trans. Signal
  Process.}, vol.~49, no.~10, pp. 2335--2344, Oct. 2001.

\bibitem{R_ortho}
R.~Raich, H.~Qian, and G.~Zhou, ``Orthogonal polynomials for power amplifier
  modeling and predistorter design,'' \emph{{IEEE} Trans. Veh. Technol.},
  vol.~53, no.~5, pp. 1468--1479, Sep. 2004.

\bibitem{K_theo}
K.~Komatsu, Y.~Miyaji, and H.~Uehara, ``Theoretical analysis of in-band
  {Full}-duplex radios with parallel {Hammerstein} self-interference
  cancellers,'' \emph{{IEEE} Trans. Wireless Commun.}, vol.~20, no.~10, pp.
  6772--6786, Oct. 2021.

\bibitem{NR}
\emph{Study on Channel Model for Frequencies From 0.5 to 100 {GHz}}.\hskip 1em
  plus 0.5em minus 0.4em\relax document 3GPP TR 38.901, Oct. 2009.

\bibitem{NR2}
\emph{{NR}; Physical channels and modulation}.\hskip 1em plus 0.5em minus
  0.4em\relax document 3GPP TR 38.211, Mar. 2019.

\bibitem{SphWave}
J.-S. Jiang and M.~Ingram, ``Spherical-wave model for short-range {MIMO},''
  \emph{{IEEE} Trans. Commun.}, vol.~53, no.~9, pp. 1534--1541, Sep. 2005.

\bibitem{effSINR1}
R.~Giuliano and F.~Mazzenga, ``Exponential effective {SINR} approximations for
  {OFDM/OFDMA}-based cellular system planning,'' \emph{{IEEE} Trans. Wireless
  Commun.}, vol.~8, no.~9, pp. 4434--4439, Sep. 2009.

\end{thebibliography}

\end{document}